\documentclass[11pt, a4paper, leqno]{amsart}

\usepackage{amsxtra}
\usepackage{newlfont}
\usepackage{amscd}
\usepackage{hyperref}
\usepackage{amsmath,amsthm}
\usepackage{amssymb}
\usepackage{enumerate}
\usepackage{color}

 \newtheorem{thm}{Theorem}[section]

  \newtheorem{setup}[thm] {Setup}

 \newtheorem{lem}[thm]{Lemma}
 \newtheorem{prop}[thm]{Proposition}
 
  \theoremstyle{definition}
 \newtheorem{defn}[thm]{Definition}
  
  \newtheorem{defn-thm}[thm]{Definition-Theorem}
   \newtheorem{ex}[thm]{Example}

 \theoremstyle{remark}

 \newtheorem{rem}[thm]{Remark}

\numberwithin{equation}{section}

\numberwithin{thm}{section}

\numberwithin{table}{section}

\numberwithin{figure}{section}

\newcommand{\ZZ}{\mathbb{Z}}

\newcommand{\CC}{\mathbb{C}}

\newcommand{\II}{\mathcal{I}}

\newcommand{\Cdiff}{\CC_{\text{diff}}}

\newcommand{\g}{\mathfrak{g}}
\newcommand{\h}{\mathfrak{h}}
\newcommand{\m}{\mathfrak{m}}
\newcommand{\n}{\mathfrak{n}}
\newcommand{\p}{\mathfrak{p}}

\newcommand{\pa}{{\text{p}}}

\newcommand{\PP}{\mathcal{P}}

\newcommand{\ad}{\text{ad}}

\newcommand{\sll}{\text{sl}}
\newcommand{\gll}{\text{gl}}
\newcommand{\WW}{\mathcal{W}}

\begin{document}

\title{ Classical affine W-superalgebras via generalized Drinfeld-Sokolov reductions and related integrable systems  }
\author{ Uhi Rinn Suh$^{1}$}
\address{Department of Mathematical Science, KAIST, 291 Daehak-ro, Yuseong-gu, Daejeon 34141, South Korea}
\email{uhrisu@kaist.ac.kr}
\thanks{$^{1}$This work was supported by NRF Grant \# 2016R1C1B1010721.}
\thanks{This will appear in Communications in Mathematical Physics, accepted on Sep/06/2017  }
\maketitle

\begin{abstract}
The purpose of this article is to investigate relations between W-superalgebras and integrable super-Hamiltonian systems. To this end, we introduce the generalized Drinfel'd-Sokolov (D-S) reduction associated to a Lie superalgebra $\g$ and its even nilpotent element $f$, and we find a new definition of the classical affine W-superalgebra $\WW(\g,f,k)$ via the D-S reduction. This new construction allows us to find free generators of $\WW(\g,f,k)$, as a differential superalgebra, and two independent Lie brackets on $\WW(\g,f,k)/\partial \WW(\g,f,k).$ Moreover, we describe super-Hamiltonian systems with the Poisson vertex algebras theory. A W-superalgebra with certain properties can be understood as an underlying differential superalgebra of a series of integrable super-Hamiltonian systems.

\end{abstract}

\setcounter{tocdepth}{-1}

\pagestyle{plain}

\section{Introduction}\label{Sec:Introduction}

Classical affine W-algebras have been studied in the theory of integrable systems since 1980's, when Drinfeld-Sokolov (\cite{DS}) discovered relations between a finite dimensional simple Lie algebra $\g$ and a sequence of integrable systems. 

The main idea of Drinfel'd and Sokolov in \cite{DS} is considering Lax operators associated to a Lie algebra $\g$.  Precisely, such Lax operators have the form of  
\begin{equation}\label{Lax_DS}
L=\frac{\partial}{\partial x}+q(x)+\Lambda
\end{equation}
 where (i) $q$ is a differentiable function whose value is in a borel subalgebra $\n_+\oplus \h \subset \g$ (ii) $\Lambda=f+zs\in \g[z]$ for the principal nilpotent element $f\in \n_-$ and $s\in \ker(\ad \n_+)$. Here, $\n_+\oplus \h \oplus \n_-$ is a triangular decomposition of $\g$.  On the phase space  $\mathcal{F}_\g$ consisting of functions $q$ in Lax operators, Drinfel'd-Sokolov defined gauge transformations. As a consequence, the W-algebra $\WW(\g)$ associated to $\g$ was introduced as a set of gauge invariant functions. 
   
  Furthermore, a Lax operator (\ref{Lax_DS}) gives rise to a bi-Poisson structure on $\WW(\g)$, which does an important role to find related  integrable systems (\cite{DS}). A bi-Poisson structure $(\{\, ,\, \}_{K}, \{\, ,\, \}_{H})$ consists of a couple of linearly independent {\it local } Poisson brackets which involve delta distributions, that is,    
\[ \{ u(x), v(y)\}_X = \sum_{n\in \ZZ_+\cup \{0\}} W_n(y) \  \partial^n_y \delta(x-y)\ \text{ for }\ u,v,W_n\in \WW(\g), \ X=K,H \]
  satisfy {\it skew symmetries}, {\it Jacobi identities} and {\it Leibniz rules}.  After all, a systematic algorithm of getting a sequence of Hamiltonian integrable systems on $\WW(\g)$ was discovered. In this algorithm, they used the Lenard-Magri scheme and the bi-Poisson structure, i.e. there are $k_i\in \WW(\g)$  for $i\in \ZZ_{\geq 0}$ such that 
 \[ \frac{d\phi(x,t)}{dt}= \int \{ k_i(x),\phi(y)\}_H dy =  \int \{ k_{i+1}(x),\phi(y)\}_K dy, \text{ for }  \phi\in \WW(\g),\] are all distinct integrable systems.

\vskip 2mm
In an algebraic point of view, classical affine W-algebras are Poisson vertex algebras. A Poisson vertex algebra (PVA) is a differential algebra endowed with a Poisson $\lambda$-bracket structure, denoted by $\{\, _\lambda\, \}$. Here, a $\lambda$-bracket can be understood as an algebraic interpretation of a local Poisson bracket. On the other hand, Poisson vertex algebras are closely related to vertex algebras since the quasi-classical limit of a certain family of vertex algebras is a Poisson vertex algebra. As one can expect, there is a vertex algebra called a quantum affine W-algebra which is a quantization of a classical affine W-algebra.

\vskip 2mm

 The quantum affine W-(super)algebra associated to $\g$ and $f\in \g$ is introduced via the quantum affine BRST complexes (or the quantum Drinfeld-Sokolov reduction), provided that $\g$ is a finite dimensional simple {\it Lie superalgebra}  with a non-degenerate even supersymmetric bilinear form and $f$ is an even nilpotent element with an $\sll_2$-triple $(e,h,f)$ (\cite{DK,FF,KRW,KW}). In \cite{S, S2}, the author proved that the quasi-classical limit of the quantum affine W-(super)algebra associated to $\g$ and $f$ is
\begin{equation} \label{Eqn:CA_W_intro}
 \WW(\g,f,k)= \big(\PP / \mathcal{I})^{\ad_\lambda \n},
 \end{equation}
where $\PP=S(\CC[\partial]\otimes \g)$ is the affine PVA, and the Lie subalgebra $\n$ of $\g$ and the differential algebra ideal $\mathcal{I}$ of $\PP$ are determined by $f$. Here the $\ad_\lambda \n$-action is induced from the $\lambda$-bracket on the affine PVA $\PP$ and $k\in \CC$ is the constant involved in the $\lambda$-bracket of $\PP$. These results imply that classical affine W-(super)algebras have properties analogous to those of finite W-(super)algebras (\cite{GG}). Note that the W-algebra $\WW(\g)$ introduced by Drinfel'd and Sokolov in \cite{DS} is just $\WW(\g,f,k)$, when $f=f_{\text{princ}}$ is a principal nilpotent element and $k=1$ (see also \cite{DKV}).\

\vskip 2mm

Since $\WW(\g)$ is a special case of W-algebras associated to $\g$, there have been many attempts to understand algebraic structures of $\WW(\g,f,k)$ and to find integrable systems associated to $\WW(\g,f,k)$ for any nilpotent element $f$ (see \cite{BDHM, DHM} for instance).  Regarding these topics, there are  plenty of considerable articles, provided that $\g$ is a {\it Lie algebra}. In \cite{DKV1, V}, using the definition (\ref{Eqn:CA_W_intro}), De Sole, Kac and Valeri succeeded in explaining generators of $\WW(\g,f,k)$ and the $\lambda$-bracket relations between them. Moreover, integrable systems on $\WW(\g,f,k)$ are discovered in \cite{BDK, DKV5, DKV3, DKV2}. However, in the case when $\g$ is a Lie superalgebra, there still remain many open problems.

\vskip 2mm 

For algebraic structures of W-algebras associated to Lie superalgebras (W-superalgebras), we refer \cite{Pol, PS,ZS,Zhao}, where structures of finite W-superalgebras are considered, and \cite{KW,S3}, where the algebraic structures of affine W-superalgebas associated to minimal nilpotent elements are given. On the other hand,  integrable systems has not been yet explored in precise connections with W-superalgebras (\ref{Eqn:CA_W_intro}) via PVA structures, to the best of the author's knowledge. There are some articles which investigated integrable systems on noncommutative algebras (see \cite{DKV4, IK, IK1, KZ} and the references therein). In particular, in \cite{IK, IK1, KZ}, the authors described relations between integrable systems and Lie superalgebras, for instance $\text{spo}(2|1)$ and $\sll(n|n)$. However, it is not clear if these integrable systems can be explained by PVA structures of W-superalgebras (\ref{Eqn:CA_W_intro}).

\vskip 2mm

In this context, a natural question is whether a W-superalgebra can be related to a sequence of integrable systems. In this paper, as the  first step toward answering this question, we construct W-superalgebras using Lax operators, mainly inspired by the important papers \cite{BDHM, DHM, DKV,DS}. The key idea is to consider Lax operators in algebraic languages:
\[ L=k \partial+ q-\Lambda \otimes 1 \ \in \ \CC\partial \ltimes (\g[z] \otimes  \PP/\mathcal{I})_{\bar{0}}\]
 {\it with even parities}, where $\PP/\mathcal{I}$ is the differential algebra in (\ref{Eqn:CA_W_intro}) (see Definition \ref{Def:Lax_super}).  In other words, we assume  that  the entire phase space 
 \[ \mathcal{F}_{\g,f}=\{ q\, |\, L=k \partial+q-\Lambda\otimes 1\,\text{ is even Lax operator} \}\, \subset\,  (\g[z] \otimes  \PP/\mathcal{I})_{\bar{0}}\]
is even. By considering {\it even gauge transformations}, we prove that the construction of W-superalgebras via Lax operators is equivalent to (\ref{Eqn:CA_W_intro}) (see Theorem \ref{Thm:3.11_0730}). This new construction is particularly useful to find free generators of W-superalgebras as differential algebras (see Proposition \ref{Prop:3.14_0730}). 

\vskip 2mm

Recall that if $\g$ is a Lie algebra and $f\in \g$ is a nilpotent element, then the W-algebra  $\WW(\g,f,k)$ is endowed with a pair of Poisson $\lambda$-brackets, and these brackets play crucial roles to describe integrable Hamiltonian systems associated to $\WW(\g,f,k)$. Therefore, in order to find integrable systems associated to W-superalgebras in an analogous approach, we need to understand the followings:
 \begin{itemize}
\item the definition of {\it super-Hamiltonian evolution equations} using Poisson $\lambda$-brackets of Poisson vertex algebras;
\item how to describe two linearly independent Lie (super)brackets on $\WW(\g,f,k)/\partial \WW(\g,f,k)$ with Lax operators and variational derivatives;
\item Lenard-Magri scheme associated to super-Hamiltonian integrable systems.
\end{itemize}
To describe two linearly independent Lie (super)brackets on $\WW(\g,f,k)/\partial \WW(\g,f,k)$, we consider the special operator called the {\it sign twisted\footnote{The main reason we consider {\it sign twisted} operators is explained in Remark \ref{Rem:4.16_0721}.} universal Lax operator,}
\[ L_{\text{univ}}^\sigma= k\partial +q_{\text{univ}}^\sigma-\Lambda\otimes 1 \in \CC\partial \ltimes \g[z]\otimes  \PP/\II.\]
(see Proposition \ref{Prop:signLax} for the precise definition). Employing this operator, we describe the Lie brackets on $\WW(\g,f,k)/\partial \WW(\g,f,k)$ which are needed to use the Lenard-Magri scheme. Indeed, under the assumption that $\Lambda$ is semisimple in $\g((z^{-1}))$, we show that there are super-Hamiltonian integrable systems associated to classical affine W-superalgebras (see Theorem \ref{Thm:5.10_0730}).\footnote{this is a quite strong assumption (see Remark \ref{Rem:Lambda_assump}). An important open question in this field is to find integrable systems associated to arbitrary W-superalgebras.}  As a simplest example, we show one of super-Hamiltonian integrable systems associated to $\WW(\text{spo}(2|1), f,k)$ is equivalent to the super KdV equation, which appears in \cite{KZ}, up to constant factors.

\section*{Acknowledgement}

The author would like to thank the Institut des Hautes \'Etudes Scientifiques (IHES) for their hospitality, where a part of this work was done. The author also wishes to thank Prof. Victor Kac for helpful discussion, and thank anonymous referees who pointed out mistakes and suggested various ways to improve this paper.  

\section{Classical affine W-algebras}

In this section, we review some known facts about classical affine W-algebras. For the notions about Poisson vertex algebras, we refer to \cite{BK,K}. Properties about W-algebras and Lax operators denoted in this section can be found in \cite{DKV,DS}.

\subsection{Poisson vertex algebras}\ 

A {\it vector superspace} is a vector space $V$ with the $\ZZ/2\ZZ$-graded decomposition $V=V_{\bar{0}}\oplus V_{\bar{1}}.$ For $i=0,1$, we denote the {\it parity} by $\text{p}(a)=i$ for a homogeneous element $a\in V_{\bar{i}}.$ The {\it superalgebra} $\text{End}(V)=\text{End}(V)_{\bar{0}} \oplus\text{End}(V)_{\bar{1}}$ is a vector superspace such that \[ \text{ if } F\in \text{End}(V)_{\bar{i}} \text{ then } F(V_{\bar{j}})\subset V_{\bar{i}+\bar{j}}.\]
The {\it supersymmetric algebra} $A$ is a superalgebra with the {\it supersymmetry}
\[ ab=(-1)^{\pa(a)\pa(b)} ba,\]
for homogeneous elements $a$ and $b$.
\begin{defn}
A {\it Lie conformal algebra} (LCA) $R$ is a $\CC[\partial]$-module with a $\CC$-linear $\lambda$-bracket  
\[ \ [\, _\lambda\, ]\, : \, R\otimes R \to R[\lambda] \, \]
satisfying the following properties:
\begin{itemize}
\item (sesquilinearity) $[a_\lambda \partial b]=(\lambda+\partial)[a_\lambda b]$, $[\partial a_\lambda b]= -\lambda[a_\lambda b],$ 
\item (skewsymmetry) $[a_\lambda b]= -(-1)^{\pa(a)\pa(b)}[b_{-\lambda-\partial} a]$,
\item (Jacobi identity) $[a_\lambda[b_\mu c]]=[[a_\lambda b]_{\lambda+\mu} c]+(-1)^{\pa(a)\pa(b)}[b_\mu [a_\lambda c]].$
\end{itemize}
Here,  we assume $\partial$ is an even operator on $R$.
\end{defn}

\begin{rem} Let $R$ be a LCA.
\begin{enumerate}
\item The sesquilinearity implies $\partial$ is a derivation  for the $\lambda$-bracket on $R$, i.e.,  $ \partial [a_\lambda b]= [\partial a_\lambda b]+[a_\lambda \partial b].$
\item  For any $a,b\in R$, we denote by $[a_\lambda b] = \sum_{n\geq 0}\frac{a_{(n)}b}{n!} \lambda^n$ for $a_{(n)}b \in R$.  Here,  $a_{(n)}b\in R$ is called the $n$-th product of $a$ and $b$.

\end{enumerate}
\end{rem}

\begin{defn}\ 
\begin{enumerate}
\item A {\it differential (super)algebra} $D$ is an associative (super)algebra with an even operator $\partial: D \to D$ called {\it derivation} such that 
\begin{equation} \label{Eqn:leib}
 \partial (AB)= \partial(A) B + A \partial(B) \quad \text{for} \ A, B\in D.
 \end{equation}
In other words, the derivation $\partial$ defined on a generating set of $D$ can be extended using the Leibniz rule (\ref{Eqn:leib}).
\item The {\it (super)algebra of differential  polynomials} \[\Cdiff [\, w_i\, |\, i\in I_{\bar{0}}\cup I_{\bar{1}}\, ]\] generated by even elements in $\{w_i\}_{i\in I_{\bar{0}}}$ and odd elements $\{w_i\}_{\i\in I_{\bar{1}}}$ is the algebra isomorphic to the tensor products of a symmetric algebra and a exterior algebra
\[S( V_0) \otimes \bigwedge(V_1 ) \]
where $V_0:= \text{Span}_{\CC}(\partial^n\,  w_i\, | \, i\in I_{\bar{0}}, \, n\in \ZZ_{\geq 0})$ and $V_1:= \text{Span}_{\CC}(\partial^n\,  w_i\, | \, i\in I_{\bar{1}}, \, n\in \ZZ_{\geq 0}).$

\end{enumerate}
\end{defn}

\begin{rem}
For the simplicity of notations, we denote by $S(V)$ the supersymmetric algebra generated by the vector superspace $V=V_{\bar{0}}\oplus V_{\bar{1}}.$ In other words,
\[S(V):= S(V_{\bar{0}}) \otimes \bigwedge(V_{\bar{1}}).\]
\end{rem}

Using the notion of a LCA and a differential algebra, we can introduce Poisson vertex algebras.

\begin{defn}
A quintuple $(\mathcal{P}, 1, \{\, _\lambda \, \}, \partial, \cdot)$ is a {\it Poisson vertex algebra }(PVA) if  \\
(1) $(\mathcal{P},  \{\, _\lambda \, \}, \partial)$ is a Lie conformal algebra, \\
(2) $(\mathcal{P}, 1, \cdot, \partial)$ is a supersymmetric differential algebra with the derivation $\partial$ and the unity $1$, \\
(3) the Leibniz rule holds: \[ \ \{A_\lambda BC\}=(-1)^{\pa(A)\pa(B)} B\{A_{\lambda} C\} + (-1)^{\pa(C)(\pa(A)+\pa(B))} C\{A_{\lambda} B\}.\] 
\end{defn}

\begin{ex} \label{Ex:affine PVA}
Let $\g$ be a finite simple Lie superalgebra with the even supersymmetric bilinear form $(\, |\,)$. The {\it affine LCA} of $\g$ is $R=\CC[\partial]\otimes \g$ with the $\lambda$-bracket defined by 
\[ [a_\lambda b]=[a,b]+\lambda k (a|b), \quad \text{ for }  a,b\in \g \text{ and }   k\in \CC,\]
and sesquilinearity. The {\it affine PVA} of $\g$ is the (super)symmetric algebra $S(R)$ generated by $R$ endowed with the $\lambda$-bracket induced from the bracket of $R$ and the Leibniz rule.
\end{ex}

\begin{prop}
Let $\PP$ be a PVA and $\partial \PP$ be the subspace $\{\, \partial p\, |\,  p\in \PP\}$ of $P$. Then the quotient space $\PP/\partial \PP=\{ p+\partial \PP| p\in \PP\}$ endowed with the bracket 
\[\, [a+\partial \PP, b+\partial \PP]:= \{a_\lambda b\}|_{\lambda=0} +\partial \PP \text{ for } a,b\in \PP \]
is a well-defined Lie superalgebra.
\end{prop}

\begin{proof}
By the sesquilinearity of $\lambda$-brackets, we can see $[\partial a , b ]$ and $[a, \partial b]$ are in $\partial \PP$. Hence it  is a well-defined bilinear map $\PP/\partial \PP \times \PP/\partial \PP \to  \PP/\partial \PP$. Skew-symmetry and Jacobi identity of $[\, ,\, ]$ follow from those properties of $\{\, _\lambda \, \}.$
\end{proof}

\begin{defn}
Let $\PP$ be a PVA and $H:\PP\to \PP$ be a diagonalizable operator. Denote by $\Delta_a$ the eigenvalue of homogenous element $a\in \PP$ with respect to the operator $H$. If 
\[\Delta_1=0,\quad \Delta_{\partial a}= \Delta_a+1, \quad \Delta_{ab}=\Delta_a+\Delta_b, \quad \Delta_{a_{(n)}b}= \Delta_a+\Delta_b-n-1,\]
for any homogenous elements $a,b\in \PP$ then $H$ is called a {\it Hamiltonian operator} and $\Delta_a$ is called the {\it conformal weight} of $a$.
\end{defn}

\begin{rem}
Let $L$ be an element of a PVA $\PP$. If (i) $L_{(0)}=\partial$, (ii) $L_{(1)}$ is a diagonalizable, and  
\[ \text{(iii)}\ \{L_\lambda L\}= (\partial+2\lambda)L + c \lambda^3, \quad \text{ for } c\in \CC\]
then $L$ is called an energy momentum field of $\PP$. By sesquilinearities and Leibniz rules, $L_{(1)}$ is a Hamiltonian operator.  By convention, we denote by 
\[ L_{n}=L_{(n+1)}\text{ for } n\geq -1.\]
The operator $L_0=L_{(1)}$ is called the Hamiltonian operator induced from the energy momentum field $L$. 
\end{rem}

Hamiltonian operators are useful to describe relations between PVAs (VAs) and Poisson algebras (associative algebras).


\subsection{Classical affine W-algebras}\ 

There are some ingredients to construct a classical affine W-algebra. From now on, we fix notations to indicate them.

\begin{setup} \label{Setup}
Let $\g$ be a finite simple Lie superalgebra and $(e,h,f)$ be an  even $\sll_2$-triple in $\g$. Suppose $(\, | \, )$ is the even supersymmetric bilinear form on $\g$ such that $(e|f)= \frac{1}{2}(h|h)=1$ and $\g=\bigoplus_{i=-d}^{d}\g(i)$ is the $\ad\frac{h}{2}$ eigenspace decomposition. We consider two Lie subalgebras $\n$ and $\m$:
\[ \textstyle \n = \bigoplus_{i>0} \g(i)= \bigoplus_{i\geq 1/2} \g(i)  \quad  \supset \quad  \m= \bigoplus_{i\geq 1} \g(i).\]
Recall that the grading by $\ad h$ is called a {\it Dynkin grading}, which is an example of a {\it good grading}. Hence we have the following properties:
\begin{enumerate}
\item $\ad f\, : \bigoplus_{i\geq 1/2} \g(i) \to   \bigoplus_{i\geq -1/2}\g(i)$ is injective,
\item $\ad f\, : \bigoplus_{i\leq 1/2}\g(i)\to \bigoplus_{i\leq -1/2}\g(i)$ is surjective.
\end{enumerate}
By (1) and (2),  we have the bijection $\ad f\, : \bigoplus_{i = 1/2} \g(i) \to   \bigoplus_{i = -1/2}\g(i)$.
\end{setup}

Using notations in Setup \ref{Setup}, we define classical affine W-algebras.

\begin{defn} \label{Def:2.11_0801}
Let $\mathcal{I}$ be the differential algebra ideal of the affine PVA $\mathcal{P}=S(\CC[\partial]\otimes \g)$  generated by $m+(f|m)$ for  $m\in \m$.
The {\it classical affine W-(super)algebra $\WW(\g,f,k)$ associated to $\g$, $f$ and $k\in \CC$} is 
\[ \WW(\g,f,k)= \big( \mathcal{P}/ \mathcal{I} \big)^{\ad_\lambda \n},\]
where $\ad_\lambda \n$-action on $\mathcal{P}/ \mathcal{I} $ is induced from the $\lambda$-bracket on $\mathcal{P}$ in Example \ref{Ex:affine PVA}. The W-algebra $\WW(\g,f,k)$ is a PVA with the $\lambda$-bracket induced from  that of $\mathcal{P}.$ 
\end{defn}

\begin{prop}[\cite{DKV}]\label{Prop:second}
Suppose there is an even element $s\in \g(d)$, where $d$ is the largest integer such that $\g(d) \neq \{0\}.$ (See Remark \ref{Rem:s}.)
The W-(super)algebra $\WW(\g,f,k)$ in Definition \ref{Def:2.11_0801} is endowed with another Poisson $\lambda$-bracket which is induced from the bracket on the affine PVA $\mathcal{P}:=S(\CC[\partial]\otimes \g)$ defined by 
\begin{equation}\label{second bracekt}
 \{ a_\lambda b\}_2= (s|[a,b]).
 \end{equation}
\end{prop}

\begin{proof}
Consider the one parameter family of Poisson $\lambda$-brackets on $S(\CC[\partial]\otimes \g)$ defined by 
\begin{equation} \label{one_para_bracket}
\{ a_\lambda b\}^t= [a,b]+k\lambda(a|b) + t(s|[a,b]), \quad t\in \CC.
\end{equation}
Observe that  $[s,n]=0$ and $\{n_\lambda A\}=\{n_\lambda A\}^0$, where $\{\, _\lambda \, \}$ is the Poisson $\lambda$-bracket on the affine PVA in Example \ref{Ex:affine PVA}. Hence $\{n_\lambda A\}^t=\{n_\lambda A\}$. Thus, for the ideal $\mathcal{I}$ in Definition \ref{Def:2.11_0801}, we have  
\[ \WW^t(\g, f, k):= \{ \bar{A} \in \mathcal{P}/\mathcal{I}\,  | \, \{ n_\lambda A\}^t\in \mathcal{I}  \} \, \simeq \, \WW(\g,f, k), \quad A\mapsto A,\]
as differential algebras. One can also check that $\WW^t(\g,f,k)$ is a PVA endowed with the $\lambda$-bracket induced from (\ref{one_para_bracket}) and extended via Leibniz rules. For $A,B\in \WW(\g,f,k)$, we have 
\[ \{A_\lambda B\}_2 :=\{A_\lambda B\}^{t+1}-\{A_\lambda B\}^t \in \WW(\g,f,k)\]
which defines another $\lambda$-bracket on $\WW(\g,f,k)$. The well-definedness of the bracket $\{\, _\lambda \}_2$ can be shown by the master formula, or Proposition \ref{Prop:master formula}.
This bracket can be understood as the bracket induced from (\ref{second bracekt}).
\end{proof}

\begin{rem}
In order to distinguish two $\lambda$-brackets on $\WW(\g,f,k)$, we denote by $\{\, _\lambda \, \}_1$ or $\{\, _\lambda \, \}$ the bracket in Definition \ref{Def:2.11_0801} and by $\{\, _\lambda \, \}_2$ the bracket in Proposition \ref{Prop:second}.
\end{rem}

\begin{rem}\label{Rem:s}
A nonzero even element  $s\in \g(d)$ exists for a subalgebra of $\gll(m|n)$. In \cite{H}, Hoyt showed that Dynkin grading on $\g_{\bar{0}}$ can be extended to $\g.$ For example, in $\sll(m|n)$ case, it can be shown as follows.  A Dynkin grading of $\gll(m|n)$ corresponds to a pair $(\lambda|\mu)$ of partitions of $m$ and $n$. If $\lambda=(p_1, p_2, \cdots, p_{r_0})$ and $\mu=(q_1, q_2, \cdots, q_{r_1})$ are decreasing sequences then the largest numbers $d_{\bar{0}}$ and $d_{\bar{1}}$ such that $\g_{\bar{0}}(d_{\bar{0}})\neq \{0\}$ and $\g_{\bar{1}}(d_{\bar{1}})\neq \{0\}$ satisfy
\[ d_{\bar{0}}= \text{max}\{2(p_1-1), \, 2(q_1-1)\}, \qquad d_{\bar{1}}= (p_1-1)+(q_1-1).\]
Hence, for $d=\text{max}\{d_{\bar{0}},\, d_{\bar{1}}\}=d_{\bar{0}}$, there exists a nonzero even element $s\in \g(d).$ Similar argument works in $\text{spo}(m|n)$ cases.
\end{rem}

\begin{prop} [\cite{DKV}]\label{Prop:2.11_0731}
Let $\{u_i\}_{i\in I}$ and $\{u^i\}_{i\in I}$ be dual bases of $\g$ with respect to the bilinear form $(\cdot |\cdot)$ and let $\{v_i\}_{i\in I_{1/2}}$ and $\{v^i\}_{i\in I_{1/2}}$ be the dual bases of $\g_{1/2}$ with respect to the bilinear form $\omega(\cdot|\cdot)$ on $\g_{1/2}$ defined by 
\[ \omega(a,b)= (f|[a,b]). \]
 Then 
 \begin{equation} \label{Eqn:affine energy momentum}
 L= \sum_{i\in I}\frac{1}{2k}u^i \, u_i + \sum_{i\in I_{1/2}} \frac{1}{2} \partial (v^i)\, v_i+\frac{\partial h}{2}
 \end{equation}
is an energy momentum field of the affine PVA of $\g$ with the $\lambda$-bracket 
\[ \{a_\lambda b\}= [a,b]+k \lambda (a|b), \quad \text{for} \quad  a,b\in \g.\]
The conformal weight $\Delta_a$ of $a\in \g(j_a)$ is $1-j_a$. Moreover, the element $L\in \WW(\g,f,k)$ which is the quotient of $L$ in (\ref{Eqn:affine energy momentum}) is an energy momentum field of $\WW(\g,f,k).$
\end{prop}



About algebraic structures of classical affine W-algebras, the following proposition can be found in \cite{S3}. Also we note that, in \cite{DK}, there is the analogous result for quantum affine W-algebras.

\begin{prop}[\cite{DKV,S3}] Let $L$ be the energy momentum field of $\WW(\g,f,k)$  in (\ref{Eqn:affine energy momentum}). For the Hamiltonian operator $L_0$, let $\Delta_a$ be the conformal weight of the homogenous element $a\in \WW(\g,f,k).$ 
\begin{enumerate}
\item As differential algebras  
$ \WW(\g,f,k)\simeq S(\CC[\partial]\otimes \g_f),$
where $\g_f=\ker (\ad\, f)\subset \g.$ 
\item Let $B_f$ be a basis of $\g_f$.  There is $\Phi_f=\{\, \phi_g\, |\, g\in B_f \, \}\subset \WW(\g,f)$ such that 
\[ \phi_g= g+\psi_g, \]
where $\Delta_{\psi_g}=\Delta_g$ and $\psi_g$ is an element of the algebra of  differential polynomials generated by $\bigoplus_{i>1-\Delta_g} \g(i)$. Note that $g\in \g(1-\Delta_g).$ Moreover, 
\[ \WW(\g,f,k)\quad = \quad  \Cdiff[\phi_g | g\in B_f].\]
\end{enumerate}
\end{prop}

\begin{ex} \label{Ex:W(sl2)}
Let $\g=\sll_2$. Then $\n=\m= \CC e$ and the ideal $\II$ of $S(\CC[\partial]\otimes \g)$ is generated by $e+1$. As differential algebras, 
\[ \WW(\g,f,k)\quad = \quad \Cdiff[ \phi_f],\]
where $\phi_f= f-\frac{1}{2} x^2-k\partial x$ for $x=\frac{h}{2}.$ We can check that 
\[ \{ \phi_{f\ \lambda}\, \phi_f\}= -k(\lambda+2\partial) \phi_f -\frac{k^3}{2} \lambda^3.\]
\end{ex}

\begin{ex}\label{Ex:W(spo(2,1))}
Let $\g=\text{spo}(2|1)\subset \mathfrak{gl}(2|1).$ Then the even part $\g_{\bar{0}}$ is generated by an $\sll_2$-triple $(e_{ev}, h, f_{ev})$ and the odd part  $\g_{\bar{1}}$ is generated by $e_{od}$ and $f_{od}$. As matrix forms,  
\begin{equation*} 
h= \left(
\begin{array}{ccc}
1& 0& 0\\
0 &-1 &0 \\
0 & 0 & 0
\end{array}\right),
\qquad
e_{ev}= \left(
\begin{array}{ccc}
0 & 1& 0\\
0 & 0 &0 \\
0 & 0 & 0
\end{array}\right),
\qquad
f_{ev}= \left(
\begin{array}{ccc}
0 & 0 & 0\\
1 & 0 &0 \\
0 & 0 & 0
\end{array}\right),
\end{equation*}

\begin{equation*}
e_{od}= \left(
\begin{array}{ccc}
0 & 0 & 1\\
0 & 0 &0 \\
0 & 1 & 0
\end{array}\right),
\qquad
f_{od}= \left(
\begin{array}{ccc}
0 & 0 & 0\\
0 & 0 &1 \\
-1 & 0 & 0
\end{array}\right) .
\end{equation*} 

Consider the even supersymmetric invariant bilinear form $(\, | \, )$ such that $(h|h)=2(e_{ev}|f_{ev})=2$ and $(e_{od}|f_{od})=-2.$ There are two elements 
\[\phi_{od}:=f_{od}-\frac{1}{2}e_{od}h-k\partial e_{od}, \quad \phi_{ev}:=f_{ev}+\frac{1}{2}f_{od}e_{od}-\frac{1}{4}h^2+k\frac{1}{4}e_{od}\partial e_{od}-k\frac{1}{2}\partial h,\] which satisfy 
\[ \ad_\lambda e_{ev} (\phi_{od})=\ad_\lambda e_{od} (\phi_{od})= \ad_\lambda e_{ev} (\phi_{ev})=\ad_\lambda e_{od} (\phi_{ev})=0+I .\]
Hence  
\[ \WW(\g,f_{ev},k)=\Cdiff[\phi_{od},\phi_{ev}]\]
 as a differential algebra. 
By direct computations, we can check that the $\lambda$-bracket of $\WW(\g, f_{ev},k)$ is defined as follows:
\begin{equation*}
\begin{aligned}
& \{\phi_{od}\, _\lambda\, \phi_{od}\}=-2 \phi_{ev}-2k^2\lambda^2 , \\
& \{\phi_{ev}\, _\lambda\, \phi_{od}\}= -k(\partial+\frac{3}{2}\lambda)\phi_{od},\\
&\{\phi_{ev}\, _\lambda\, \phi_{ev}\}=-k(\partial+2\lambda)\phi_{ev}-\frac{k^3}{2}\lambda^3.
\end{aligned}
\end{equation*}

\end{ex}

\vskip 2mm

\subsection{Generalized Drinfeld-Sokolov reductions} \label{Subsec: D-S}\label{Subsec:D-S} \ 

In Section \ref{Subsec: D-S}, we recall the construction of  classical W-algebras associated to Lie algebras via Drinfeld-Sokolov reductions in PVA theories. For the purpose, we assume $\g$ is a finite dimensional simple Lie algebra (without odd part), in this subsection.

For a symmetric differential algebra $\mathcal{V}$, the vector space $\g\otimes \mathcal{V}$ is the Lie algebra endowed with the bracket
\[ [a\otimes F, b\otimes G]= [a,b]\otimes FG \quad \text{ for }  a,b\in \g, \ F,G \in \mathcal{V}.\]
The bilinear form $(\, | \, ): \g \times \g \to \CC$ on $\g$ can be extended to the map $(\, | \, ): \g \otimes \mathcal{V} \times \g \otimes \mathcal{V} \to \mathcal{V}$  by 
\[ ( \, a\otimes F \, | \, b\otimes G\, ) = (a|b) FG. \] 
The derivation $\partial: \mathcal{V} \to \mathcal{V}$ on $\mathcal{V}$  can be extended to an endomorphism on $\g\otimes \mathcal{V}$  such that
\[ \partial ( a\otimes F  )= a \otimes \partial F . \]
 Consider $\CC\partial$ as the trivial one dimensional Lie algebra. Then $\CC\partial \ltimes \g \otimes \mathcal{V}$ is the semidirect product of Lie algebras $\CC\partial$ and $\g \otimes \mathcal{V}$ endowed with the bracket 
 \begin{equation} \label{Eqn:Semidirect}
 [\, c_1\partial + a\otimes F, \, c_2 \partial+ b\otimes G\, ] = c_1(b\otimes \partial G)-c_2(a\otimes \partial F) + [a\otimes F, b\otimes G].
 \end{equation}

\vskip 3mm

Recall the notation $\m =\bigoplus_{i\geq 1} \g(i)$ in Setup \ref{Setup}. If we denote by $V^\perp:= \{\,  w \in \g\, |\, (w|v)=0 \text{ for any } v\in V\, \}$ for a subset $V\subset \g$ then 
\begin{equation}\label{Eqn:mperp}
\textstyle \m^\perp = \bigoplus_{i>-1} \g(i)= \bigoplus_{i\geq -1/2} \g(i).
 \end{equation}
Let us consider the subspace 
\begin{equation} \label{Eqn:p}
\textstyle \p= \bigoplus_{i<1} \g(i) =\bigoplus_{i\leq 1/2}\g(i)
\end{equation}
 of $\g$.  Then we have 
 \[ \text{ (i) $\g= \m \oplus \p$, \quad   (ii)  $\p \simeq \m^\perp$ by the bilinear form $(\, | \, )$}. \]

\begin{defn} Let $\p$ be defined as (\ref{Eqn:p}) and let $\mathcal{V}(\p):=S(\CC[\partial]\otimes \p)$. 
\begin{enumerate}
\item  Let $\mathcal{F}_{\g,f}$ be the set of elements
\begin{equation}\label{Eqn:q}
q= \sum_{i\in I_{\p}} q_i \otimes P^i \in \m^\perp \otimes \mathcal{V}(\p)
\end{equation}
 for a basis  $\{\, q_i\, |\, i\in I_{\p}\, \}$ of $\m^\perp$ and a subset $\{\, P^i\, | \, i\in I_{\p}\, \}\subset  \mathcal{V}(\p)$. The set $\mathcal{F}_{\g,f}$ is called the phase space associated to $\g$ and $f$.

\item For a given $q$ and $k\in \CC$, the operator $L$ of the form 
\begin{equation} \label{Eqn:UniversalLax}
L= k\partial+q- f\otimes 1\in \CC \partial \ltimes \g \otimes \mathcal{V}(\p).
\end{equation}
is called a  {\it Lax operator} 

\item
The {\it gauge transformation} of $ q\in\mathcal{F}_{\g,f}$ with $A\in \n \otimes \mathcal{V}(\p)$ is $q^A\in \mathcal{F}_{\g,f}$ where
\[  e^{\ad A}( k\partial +q -f\otimes 1) = k \partial+q^A-f\otimes 1.\]  
On the other hand, for $q'\in \mathcal{F}_{\g,f}$, if there is $A\in \n \otimes \mathcal{V}(\p)$ such that $q'= q^A$ then we say they are {\it gauge equivalent} and write $q \sim q'$. Note that it is not hard to check that the equivalence relation is well-defined in $\mathcal{F}_{\g,f}.$
 
\end{enumerate}
\end{defn}

Elements in the differential algebra $\mathcal{V}(\p)$
can be identified with functions $\mathcal{F}_{\g,f}\to \mathcal{V}(\p)$ as follows:
\begin{equation} \label{eqn:functional_even}
p( a\otimes F) = (p|a) F, \quad PQ(a\otimes F)= P(a\otimes F) Q(a\otimes F), \quad \partial P(a\otimes F) = \partial (P(a\otimes F) ),
\end{equation}
for $p\in \p$, $P,Q \in \mathcal{V}(\p)$ and $a\otimes F \in  \mathcal{F}_{\g,f}.$ An element $P\in \mathcal{V}(\p)$ is said to be a {\it gauge invariant function} if $P(q)=P(q')$ whenever $q\sim q'$ for $q,q'\in  \mathcal{F}_{\g,f}.$

\begin{prop}[\cite{DKV}]\label{Prop:2.18_0731}
The set $\WW$ of gauge invariant functions  in $\mathcal{V}(\p)$ is a differential subalgebra of $\mathcal{V}(\p).$ Moreover, $\WW$ is isomorphic to the classical affine W-algebra $\WW(\g,f,k)$ associated to $\g$ and $f$ as differential algebras. 
\end{prop}

\begin{rem}
In \cite{DS}, a pair of local Poisson structures in $\WW$ is described by a Lax operator. The Poisson structures are equivalent to the PVA structures on the classical affine W-algebra $\WW(\g,f,k)$ which are induced from those in the affine PVA $S(\CC[\partial]\otimes \g)$. (See Definition \ref{Def:2.11_0801} and Proposition \ref{Prop:second}.)
\end{rem}

The construction of W-algebras in Proposition \ref{Prop:2.18_0731} allows to compute generators of the algebras (see Theorem \ref{Thm:2.17_0730}).

\begin{lem}[\cite{DKV}] \label{Lem:2.15_1227}
Let $V$ be a subspace of $\m^\perp$ such that $\m^\perp= [\n,f] \oplus V.$ Take a basis $\{ v_i \}_{i\in I_{\p}}$ of $\m^\perp$ such that $\{v_i\}_{i\in  J\subset I_{\p}}$ is a basis of $V$ and $\{ v_i \}_{i\in I_{\p}\backslash J}$ is a basis of $[\n, f]$.  If $\{v^i\}_{i\in I_{\p}}$ is the dual basis of $\p$ then $\{v^i\}_{i\in J}$ is a basis of $\g^f:= \text{ker}(\ad f)$.
\end{lem}

\begin{thm}[\cite{DKV}]\  \label{Thm:2.17_0730} 
\begin{enumerate}
\item Let $\{q_i\}_{i\in I_{\p}}$ and $\{q^i\}_{i\in I_{\p}}$ be bases of $\m^\perp$ and $\p$ such that $(q_i| q^j)= \delta_{ij}.$ Denote by
\begin{equation}
\textstyle q_{\text{univ}}= \sum_{i\in I_{\p}} q_i \otimes q^i   \quad \text{ and }\quad   \mathcal{L}=k\partial+ q_{\text{univ}} -f \otimes 1.
\end{equation}
Then there is unique $X \in \n \otimes \mathcal{V}(\p)$ such that $q^X_{\text{univ}}\in V\otimes \mathcal{V}(\p)$ satisfies 
\begin{equation}
e^{\ad X} \mathcal{L}= k\partial+q^X_{\text{univ}}-f\otimes 1.
\end{equation}
\item As in Lemma \ref{Lem:2.15_1227}, let $\{q_i\}_{i\in J\subset I_{\p}}$ be a basis of $V$. If 
\[  \textstyle q^X_{\text{univ}}= \sum_{i\in J} q_i \otimes w_i\]
for $q_{\text{univ}}^X$ in (1) then $w_i$ are gauge invariant functions in  $\mathcal{V}(\p)$. Moreover, by Lemma \ref{Lem:2.15_1227}, we have $w_i= v^i +(\text{degree}\geq 2\text{ part}).$
\item The set of gauge invariant functions in $\mathcal{V}(\p)$ is the  algebra of differential polynomials
\[ \Cdiff[ \,  w_i \, |  \, i\in J\, ]. \ \]
\end{enumerate}
\end{thm}

\begin{rem} We have the differential algebra isomorphism $\mathcal{V}(\p)\simeq S(\CC[\partial]\otimes \g)/\mathcal{I}=: \mathcal{V}_\II(\p)$, $A\mapsto \bar{A}$, where $\mathcal{I}$ is the ideal defined in Definition \ref{Def:2.11_0801}. Due to the isomorphism, 
\begin{enumerate}
\item we can consider a Lax operator $L$ an element in $\CC\partial \ltimes \m^\perp \otimes \mathcal{V}_{\II}(\p),$
\item since the W-algebra $\WW(\g,f,k)$ is a subalgebra of $\mathcal{V}_\II(\p)$, we prefer to regard $\WW(\g,f,k)$ as a set of functions from $\mathcal{F}_{\g,f}$ to $\mathcal{V}_{\II}(\p).$
\end{enumerate}
\end{rem}

\begin{thm}[\cite{DKV}]
The W-algebra $\WW(\g,f,k)$ is the set of gauge invariant functions in $\mathcal{V}_\II(\p).$ Hence we can find free generators by Theorem \ref{Thm:2.17_0730} .
\end{thm}

The following is the simplest example of classical affine W-algebras.

\begin{ex}
Let $\g=\sll_2.$ Then $q_{\text{univ}}=e\otimes f + h\otimes x$ for $x=\frac{h}{2}$ and
\[\mathcal{L}= k\partial+q_{\text{univ}} -f\otimes 1.\]
If we take $X= e\otimes x$ then $q^X_{\text{univ}}= e\otimes ( f-x^2-k\partial x )$. Hence $\Cdiff[ f-x^2-k\partial x]$ is the set of gauge invariant functions. Indeed, we can check that $\phi_f:=f-x^2-k\partial x$ is a gauge invariant function as follows: 

Let $Y= e\otimes r$ for $r\in \mathcal{V}_{\II}(\p)$. Then $q^Y= h\otimes (x-r)+e\otimes (f-k\partial r-2rx+r^2)$ and 
\[ \phi_f(q^Y)= (f-k\partial r-2rx+r^2)-(x-r)^2-k\partial(x-r)=f-x^2-k\partial x\]
which is independent on $r.$ Also, we can check that the algebra of differential polynomials $\Cdiff[\phi_f]$ is isomorphic to the W-algebra $\WW(\g,f,k)$ in Example \ref{Ex:W(sl2)}. 
\end{ex}

\subsection{Integrable Hamiltonian systems associated to W-algebras}\

Integrable Hamiltonian systems can be investigated by Poisson vertex algebras theories (\cite{BDK}). In this subsection, we briefly review basic notions related to integrable Hamiltonian systems.

\begin{defn} \label{Def:HEviaPVA} Let $\mathcal{P}$ be an (even) algebra of differential polynomials with a PVA structure.  
\begin{enumerate}
\item An evolution equation is called a {\it Hamiltonian system} on $\mathcal{P}$ if there is $h\in \mathcal{P}$ such that  
\[ \frac{du}{dt}= \{h_\lambda u\} |_{\lambda=0}, \quad \text{ for } u\in \mathcal{P}.\]
\item Consider the quotient map $\int : \mathcal{P} \to \mathcal{P}/\partial \mathcal{P}$. The image $\int f $ of $f\in \mathcal{P}$ is called a {\it local functional.} 
\item A Hamiltonian system is called an {\it integrable system} if there are infinitely many linearly independent integrals of motion $\int h_i$, $i\in \ZZ_{\geq 0}$. Here, an integral of motion $\int h_i$ is a local functional such that $\int  \frac{ dh_i}{dt}=0.$
\end{enumerate}
\end{defn}

In the rest of this subsection, consider the Laurent series $\g((z^{-1}))$ with the Lie bracket \[ \, [a z^n, b z^m]= [a,b] z^{n+m} \text{ for }a,b\in \g.\ \] 
 Recall $\WW(\g,f,k)$ is endowed with a bi-Poisson $\lambda$-bracket which is induced from that on $S(\CC[\partial]\otimes \g)$:
\begin{equation}
\begin{aligned}
\{a_\lambda b\}_1  = [a,b]+k\lambda(a|b),\qquad \{a_\lambda b\}_2 = (s|[a,b]).
\end{aligned}
\end{equation}

\begin{rem}[Lenard-Magri Scheme]
Let $\mathcal{P}$ be a PVA with the bi-Poisson $\lambda$-bracket $(\{\, _\lambda\, \}_1, \{\, _\lambda\, \}_2).$ Suppose there is a sequence of linearly independent local functionals $\int h_i\in \mathcal{P}/\partial \mathcal{P}$, $i=0,1,2,\cdots$ such that 
\[\text{(i) $\{h_{0\, \lambda\, } \mathcal{P}\}_2|_{\lambda=0}=0$, \ \ \  (ii) $\{h_{i\, \lambda\, } p\, \}_1|_{\lambda=0}=\{h_{i+1\, \lambda\, } p\}_2|_{\lambda=0}$ for $i\geq 0$ and $p\in \mathcal{P}$.}\]
Then $\frac{du}{dt}= \{h_{i\, \lambda\,} u\}_K|_{\lambda=0}$ for $i=0,1,2,\cdots$ are Hamiltonian integrable systems.  
\end{rem}

\begin{thm}[\cite{DKV}]
Suppose $\Lambda := f+sz\in \g((z^{-1}))$ is semisimple for $s\in \ker(\ad\, \n)$.
There is a sequence of integrable systems on $\WW(\g,f,k)$ which satisfies the assumptions of Lenard-Magri scheme. More specifically, consider  
\[ \mathcal{L}(\Lambda):=k\partial+q_{\text{univ}}-\Lambda\otimes 1= \mathcal{L}-zs\otimes 1\in \CC\partial \ltimes \g((z^{-1})) \otimes \mathcal{V}_{\II}(\p)\]
 and take $h(z)\in \big(\ker(\ad\Lambda) \cap \g[[z^{-1}]]\big) \otimes \mathcal{V}_\II(\p)$ such that 
$ e^{\ad S(z)}  \mathcal{L}(\Lambda)= k\partial + h(z)+\Lambda\otimes 1$
for some $S(z)\in  \g[[z^{-1}]] \otimes \mathcal{V}_\II(\p).$ Then $h_i=(z^i\Lambda\otimes 1| h(z))$ is an element in $\WW(\g,f,k)$ and the Hamiltonian equation 
\[ \frac{du}{dt}= \{h_{i\, \lambda\, } u\}_H|_{\lambda=0}\]
is an integrable system on $\WW(\g,f,k).$ Here, the bilinear form on $\g((z^{-1})) \otimes\mathcal{V}_\II(\p)$ is defined by $(az^n \otimes F | b z^m \otimes G)= (a|b) FG \delta_{n+m, 0}$ for $a,b\in \g$ and $F,G\in \mathcal{V}_\II(\p)$.
\end{thm}

\begin{rem} (\cite{DKV}) \label{Rem:assumption}
It is a natural question to ask that if we can find a semisimple element $\Lambda= f+zs$ for a given nilpotent element $f$. In the case when $\g= \sll_n$, if $f$ corresponds to one of the following partitions $\lambda$ of $n$ then we can find such a semisimple element $\Lambda$.
\[ (1)\ \lambda=(r,r, \cdots, r, 1,1, \cdots, 1), \quad (2) \ \lambda= (r, r-1, r, r-1, \cdots, r, r-1, 1,1, \cdots, 1).\]
\end{rem}

\section{Classical affine W-superalgebras and generalized Drinfel'd-Sokolov reductions} \label{Sec:W-super}

In this section, we shall show a set of generators of a classical affine W-superalgebra as a superalgebra of differential  polynomials can be obtained by an analogous method to the generalized Drinfeld-Sokolov reduction. 

Let $\g$ be a simple Lie superalgebra with a nondegenerate invariant even supersymmetric bilinear form $(\, | \, )$ and let  $\mathcal{V}$ be a supersymmetric differential superalgebra. The vector superspace $\g \otimes \mathcal{V}$ is endowed with the Lie bracket and the bilinear form defined by 
\[ [ \, a\otimes F , \, b\otimes G]= (-1)^{\pa(b)\pa(F)} [a,b]\otimes FG,  \qquad (\, a\otimes F\, | \, b\otimes G\, )= (-1)^{\pa(b)\pa(F)} (a|b)FG \] 
for the homogeneous elements $a,b\in \g$ and $F,G\in \mathcal{V}.$

Due to the invariance of the bilinear form on $\g$, we get the invariance of the bilinear form $(\, |\, )$ on $\g\otimes \mathcal{V}$ 
\begin{equation*}
\begin{aligned}
 & (\, a \otimes F\, | \,[ b\otimes G, c\otimes H]\, )=(\, [a \otimes F, b\otimes G]\, |\, c\otimes H]\, ),
 \end{aligned}
\end{equation*}
 for $a,b,c\in \g$ and $F,G,H\in \mathcal{V}$.

\vskip 2mm

Let us consider an even derivation $\partial:\mathcal{V}\to\mathcal{V}$ on $\mathcal{V}$. Then it can be extended to the map on $\g\otimes \mathcal{V}$ by $\partial(a\otimes F)= a \otimes \partial F$. The Lie superalgebra
\[ \CC\partial\ltimes (\g \otimes \mathcal{V})\]
is the semidirect product of the trivial Lie algebra $\CC\partial$ and the Lie superalgebra $\g \otimes \mathcal{V}$.

\vskip 2mm

Suppose the Lie superalgebra $\g$ has an $\sll_2$-triple $(e,h,f)$ with the even supersymmetric bilinear form $(\, |\, )$ such that $(e|f)= \frac{1}{2}(h|h)=1$. As in Section \ref{Subsec:D-S}, let $\m=\bigoplus_{i\geq 1} \g(i)$, $\m^\perp=\bigoplus_{i>-1}\g(i)$ and $\p= \bigoplus_{i< 1} \g(i)$, where $\g(i)$ is the $\ad \frac{h}{2}$ eigenspace with the eigenvalue $i$. Recall that  $\g=\m\oplus \p$ and $\p \simeq \m^\perp$ as vector superspaces via the bilinear form $(\, |\, )$ on $\g$. For the superspace $\m^\perp=\m^\perp_{\bar{0}}\oplus \m^\perp_{\bar{1}}$, there is a basis $\{\, q_i\, |\,  i\in I:=  I_{\bar{0}}\cup I_{\bar{1}}\, \}$  of $\m^\perp$ such that 
\begin{equation}
 \text{(i) }\{\, q_i\, |\, i\in I_{\bar{0}}\, \} \text{ is a basis of  } \m^\perp_{\bar{0}}, \ \text{ (ii) } \{\, q_i\, |\, i\in I_{\bar{1}}\, \}\text{ is a basis of } \m^\perp_{\bar{1}}.
 \end{equation}



\begin{defn} \label{Def:Lax_super}
Let $\mathcal{V}(\p):= S(\CC[\partial]\otimes \p)$ be the differential superalgebra generated by the vector superspace $\p$.
A {\it Lax operator} $L$ is an even element in $\CC\partial \ltimes  \g \otimes  \mathcal{V}(\p) $ such that 
\[ \textstyle L= k\partial + \sum_{i\in I_{\bar{0}}}q_i \otimes P^i +\sum_{i\in I_{\bar{1}}} q_j\otimes P^j -f\otimes 1\in\CC \partial \ltimes \big( \g \otimes  \mathcal{V}(\p) \big)_{\bar{0}},\]
where $q= \sum_{i\in I_{\bar{0}}} q_i \otimes P^i +\sum_{i\in I_{\bar{1}}} q_j\otimes P^j \in \big( \m ^\perp \otimes \mathcal{V}(\p) \big)_{\bar{0}}$.
\end{defn}


\begin{rem}\label{Rem:3.4}
 Let $\mathcal{I}$ be the differential superalgebra ideal of $\mathcal{V}(\g)$ generated by $\{m+(f|m)|m\in \m\}.$  Denote $\mathcal{V}_\II(\p) :=\mathcal{V}(\g)/\II$. Then we can check the following facts:
 \begin{enumerate}
\item $ \mathcal{V}(\p)\simeq \mathcal{V}_\II(\p)$ 
as differential superalgebras by the canonical isomorphism $\iota: \mathcal{V}(\p) \to \mathcal{V}_{\II}(\p)$. If there is no danger of confusion then we denote $\iota(P)\in \mathcal{V}_\II(\p)$ by $P$.
\item We can regard Lax operators as elements in $\CC\partial \ltimes  \big(  \m^\perp \otimes \mathcal{V}_{\II}(\p) \big)_{\bar{0}}$ via the isomorphism $\iota$ in (1).
\end{enumerate}
\end{rem}

Recall that the W-superalgebra $\WW(\g,f,k)$ is a subset of $\mathcal{V}_{\II}(\p)$. In order to see the relation between W-superalgebras and Lax operators, we use (2) in Remark \ref{Rem:3.4}. \\

A Lax operator $L$  acts on $\g\otimes \mathcal{V}_{\II}(\p)$ by 
\[ L(a \otimes F):= [L, a\otimes F].\]
Consider the phase space 
\[ \mathcal{F}_{\g,f}:= \big(  \m^\perp \otimes \mathcal{V}_{\II}(\p) \big)_{\bar{0}}.\]
Then for any $q\in \mathcal{F}_{\g,f}$, there is the corresponding Lax operator $L= k\partial+q-f\otimes 1$. Note that, for a Lax operator $L$ and an element  $X\in \big(\n \otimes  \mathcal{V}_{\II}(\p) \big)_{\bar{0}}$, there is  $q^X \in \mathcal{F}_{\g,f}$ such that
\begin{equation}\label{Eqn:GaugeTransf}
 e^{\ad X} L =e^{\ad X} (k\partial + q -f \otimes 1)= k\partial + q^X -f \otimes 1.
 \end{equation}
 Hence $e^{\ad X} L $ is again a Lax operator. 

\begin{defn}\ 
\begin{enumerate}
\item Let $q\in \mathcal{F}_{\g,f}$ and $X\in \big(\n \otimes  \mathcal{V}_{\II}(\p) \big)_{\bar{0}}$. 
Then $q^X\in \mathcal{F}_{\g,f}$ defined as in (\ref{Eqn:GaugeTransf}) is said to be the {\it gauge transformation} of $q\in \mathcal{F}_{\g,f}$ by $X$. 
\item
For two elements $q, q'\in \mathcal{F}_{\g,f}$, if there is an element $Y\in \big(  \n \otimes \mathcal{V}_{\II}(\p) \big)_{\bar{0}}$ such that $q^Y=q'$ then we say $q$ and $q'$ are {\it gauge  equivalent} and write $q\sim q'.$ 
\item 
The {\it universal Lax operator} associated to $\g$ and $f$ is 
\begin{equation} \label{Eqn:univ}
\textstyle  \mathcal{L} = k\partial +q_\text{univ}-1\otimes f=k
 \partial +\sum_{i\in I= I_{\bar{0}}\cup I_{\bar{1}}} q_i \otimes q^i - f\otimes 1,
 \end{equation}
where $\{q_i\}_{i\in I}$ and $\{q^i\}_{i\in I}$  are bases  of $\m^\perp$ and $\p,$ such that $(q_i|q^j)= \delta_{ij}$.  
\end{enumerate}
\end{defn}

\begin{rem}
Since the bilinear form $(\,|\,)$ defined on the Lie superalgebra $\g$ is {\it even}, the universal Lax operator  in (\ref{Eqn:univ}) is even. Hence it is a Lax operator.
\end{rem}

Now we identify an element in $\mathcal{V}_{\II}(\p)$  with a linear map $\partial\ltimes (\g \otimes \mathcal{V}_{\II}(\p))_{\bar{0}}  \to \mathcal{V}_{\II}(\p)$ defined by (\ref{Eqn:functional_even+odd}) and (\ref{Eqn:functional_even+odd_extend}):
\begin{equation} \label{Eqn:functional_even+odd}
p(\partial)= 0, \quad c(q)=c, \quad  p(a\otimes F)= F(a| p )=(a|p)F,
\end{equation} 
for $c\in \CC$, $p\in \p\subset \mathcal{V}_{\II}(\p) $ and $a\otimes F \in (\g \otimes \mathcal{V}_{\II}(\p))_{\bar{0}}$. For $P, Q\in \mathcal{V}_{\II}(\p)$, we have 
\begin{equation}\label{Eqn:functional_even+odd_extend}
 PQ(a\otimes F)=P(a\otimes F) Q(a\otimes F), \quad  \partial P (a\otimes F)= \partial (P (a\otimes F)).
 \end{equation}


\begin{rem}
 If $\g$ is even then $p(a\otimes F)= (p | a)F=(a|p)F$. Hence (\ref{Eqn:functional_even+odd}) and (\ref{Eqn:functional_even+odd_extend}) define the same functions as those in (\ref{eqn:functional_even}). If $\g$ is not even and $a\otimes F$ is an even element in $\g \otimes \mathcal{V}_{\II}(\p)$ then 
\[ p(a\otimes F) = (a|p)F= (-1)^{\pa(p)}(p|a)F.\]
Here, the last equality holds since $(p|a)\neq 0$ implies $\pa(p)=\pa(a)$.  The reason we consider the definition $p(a\otimes F) := (a|p)F$ instead of $p(a\otimes F) := (p|a)F$  can be explained by the proof of Lemma \ref{Lem:basic1} and Proposition \ref{Prop:3.3_0207}.
\end{rem}

\begin{defn}
A function $P\in \mathcal{V}_{\II}(\p)$ is said to be  {\it gauge invariant } if $P(q)=P(q')$ for any gauge equivalent elements $q$ and $q'$ in $\mathcal{F}_{\g,f}$.
\end{defn}

\begin{prop}
The subset of $\mathcal{V}_{\II}(\p)$ consisting of gauge invariant functions is a differential superalgebra .
\end{prop}

\begin{proof}
It is clear that if $P,Q\in \mathcal{V}_{\II}(\p)$ are gauge invariant then $PQ$ and $\partial P$ are also gauge invariant.
\end{proof}

We note that
\begin{equation}
P(\mathcal{L})= P \text{ for } P\in \mathcal{V}_{\II}(\p)
\end{equation}
 since an element $p\in \p\subset \mathcal{V}_{\II}(\p)$ has properties $(f|p)= 0$ and $p(\partial)=0$ so that   
\[
\textstyle p(\mathcal{L})  = p(q_{\text{univ}})=\sum_{i\in I_{\bar{0}}\cup I_{\bar{1}}} q^i(q_i | p)=p.
\]
Also, a Lax operator $L= k\partial + Q -f\otimes 1$ with $Q=\sum_{i\in I} q_i \otimes Q^i  \in \mathcal{F}_{\g,f}$ satisfies  
\begin{equation}\label{Eqn:compute functions}
 P(Q)=P(L)= P(\mathcal{L})|_{q^i=Q^i} \text{ for any }P\in \mathcal{V}_\II(\p).
\end{equation}
 Here, the subscript $q^i= Q^i$ means that we substitute $q^i$ by $Q^i.$
 
 \vskip 2mm
 
 Now, the following lemma is useful to see detailed computations in the proof of Proposition \ref{Prop:3.3_0207}.

\begin{lem} \label{Lem:basic1}
Let $X= n\otimes r \in (\n \otimes \mathcal{V}_{\II}(\p))_{\bar{0}}$ and $p\in \mathcal{V}_{\II}(\p)$. Then we have
\[ p([X, \mathcal{L}])=-k \partial r (n|p)-r[n,p] \in \mathcal{V}_{\II}(\p).\]
\end{lem}

\begin{proof}
We have 
\begin{equation}
\textstyle \, [X, \mathcal{L}]= -n\otimes k\partial r + [n\otimes r, q_{\text{univ}}-f\otimes 1]   =-n\otimes k\partial r - \sum_{i\in I} [q_i, n]\otimes r  q^i - [n, f]\otimes r
 \end{equation}
 since $\pa(q_i)=\pa(p^i)$ and $\pa(n)=\pa(r).$ Hence
 \begin{equation}
 \begin{aligned}
  p([X, \mathcal{L}]) & =-k\partial r ( n|p)-  r \,  q^i  (\, [q_i,n] \, |\, p\, )+ r \, (\, [n,f]\, |\, p\, )= -k \partial r (n|p) - r [n,p] \in \mathcal{V}_{\II}(\p).
  \end{aligned}
  \end{equation}
\end{proof}

\begin{prop} \label{Prop:3.3_0207}
If $W\in \mathcal{V}_\II(\p)$ is  a gauge invariant function then $W\in \WW(\g,f,k).$
\end{prop}

\begin{proof}
Let $X=  n  \otimes  r \in (\n \otimes \mathcal{V}_{\II}(\p))_{\bar{0}}$ and $\epsilon\in \CC$. For $W\in \mathcal{V}_{\II}(\p)$, we denote 
\begin{equation} \label{Eqn:3.6_0103}
\textstyle W ( \mathcal{L}+\epsilon[X, \mathcal{L}] +\frac{1}{2}\epsilon^2 [X,[X,\mathcal{L}]]+\cdots ) = \sum_{t\geq 0} \epsilon^t P^W_t.
\end{equation}
for some $P^W_t \in \mathcal{V}_\II(\p).$

If $W$ is a gauge invariant function, we have
\begin{equation} \label{Eqn:3.5_0103}
W( e^{\ad \epsilon X} (\mathcal{L}))= W
\end{equation}
for any $\epsilon\in \CC.$  Since (\ref{Eqn:3.5_0103}) implies $P^W_t=0$ for any $t\geq 1$, it is enough to show that 
\[\text{  $P^W_1=0$ implies $W\in \WW(\g,f,k).$ }\]

 In order to show that $P^W_1=0$, let us  denote by $[ A_\lambda B]_{\II}\in\mathcal{V}_{\II} (\p) [\lambda]$, for $A,B\in \mathcal{V}(\g)$, the image of $[ A_\lambda B]\in \mathcal{V}(\g)[\lambda]$. In other words,
\begin{equation}
[ A_\lambda B]_\II=[ A_\lambda B]+  \II [\lambda].
\end{equation}
Note that $[\, _\lambda\, ]_{\II}$ induces the well-defined $\lambda$-bracket on $\WW(\g,f,k)$ since 
\begin{equation} \label{Eqn:brac_well}
 [ n_\lambda P(m+(f|m)) ]_\II= 0+\II[\lambda] \text{ \  for \ } P \in \mathcal{V}(\g)=S(\CC[\partial]\otimes \g) \text{ and } m\in \m
 \end{equation}
 so that $[n_\lambda \II ] = \II[\lambda].$
 
 By the definition of $\WW(g,f,k),$ an element $W \in \mathcal{V}_{\II}(\p)$ is in $\WW(\g,f,k)$ if and only if  $[n_\lambda \widetilde{W}]_\II=0$, where $\widetilde{W}$ is  an element in $\mathcal{V}(\p)$ such that $\iota(\widetilde{W})=W$ for the map $\iota$ in Remark \ref{Rem:3.4}. Thus, it is enough to show that  
  \begin{equation}\label{Eqn:3.8_0103}
\textstyle  P^W_1= -\sum_{t\geq 0} (\partial^t r) P^W_{1t} \quad \text{ if and only if  }\quad [n_\lambda \widetilde{W}]_{\II}= \sum_{t\geq 0} \lambda^t \cdot P^W_{1t}.
 \end{equation}
 \vskip 1mm
 
 Now, observe the following facts.
 
 \begin{enumerate}[(i)]
\item If we substitute $W$ with $p\in \p\subset \mathcal{V}_\II(\p)$ in (\ref{Eqn:3.6_0103}) then, by Lemma \ref{Lem:basic1}, we have $P^W_1= -k\partial r(n|p)-r[n,p]\in \mathcal{V}_{\II}(\p).$ On the other hand, take $\widetilde{W} \in \mathcal{V}(\p)$ such that $\iota(\widetilde{W})=p$. Since $[n_\lambda \widetilde{W}]_{\II}= [n,p]+k\lambda(n|p) \in \mathcal{V}_{\II}(\p)[\lambda]$, we have (\ref{Eqn:3.8_0103}).
\item   If $W=\partial^m p$ then $P^W_1= -\partial^m(k \partial r(n|p)-r[n,p]).$ In this case, the sesquilinearity $[n_\lambda \partial^m p]= (\lambda+\partial)^m [n_\lambda p]$ implies (\ref{Eqn:3.8_0103}).
\item  Suppose $W=BC$ for homogeneous elements $B,C\in \mathcal{V}_\II(\p)$. Since 
\[ W(e^{\ad \epsilon X} \mathcal{L})= B(e^{\ad \epsilon X} \mathcal{L}) \cdot  C(e^{\ad \epsilon X} \mathcal{L}),\]
we have $ P^W_1= B P^C_1 + P^B_1 C. $
Denote $P^B_1= \sum_{t\geq 0} (\partial^t r) P^B_{1t}$  and $P^C_1= \sum_{t\geq 0} (\partial^t r) P^C_{1t}$.  Then
\begin{equation} \label{Eqn:3.17_0802}
\textstyle P^W_1= \sum_{t\geq 0}  \big[ (-1)^{\pa(r)\pa(B)} (\partial^t r)B \, P^C_{1t} + (\partial^{t} r) P^B_{1t} \, C \big].
\end{equation}
On the other hand, by the Leibniz rule, we have
\begin{equation} \label{Eqn:3.18_0802}
\begin{aligned}
\, [n_{\lambda}\widetilde{W}]_{\II}=[n_{\lambda}\widetilde{B}\widetilde{C}]_{\II}& = (-1)^{\pa(r)\pa(B)}B[n_\lambda \widetilde{C}]_{\II} + [n_\lambda \widetilde{B}]_{\II} C \\
& = \textstyle - \sum_{t\geq 0}  \big[ (-1)^{\pa(r)\pa(B)} B \, \lambda^t P^C_{1t} + \lambda^t P^B_{1t} \, C \big] ,
\end{aligned}
\end{equation}
for $\widetilde{B}, \widetilde{C}\in \mathcal{V}(\p)$ such that $\iota(\widetilde{B})=B$ and $\iota(\widetilde{C})=C$.
It is not hard to see  (\ref{Eqn:3.17_0802}) and (\ref{Eqn:3.18_0802}) imply (\ref{Eqn:3.8_0103}). 
\end{enumerate}

By (i), (ii), (iii) and the induction, we prove the proposition.
\end{proof}

\begin{prop} \label{Prop:3.4_0211}
If $W\in \WW(\g,f,k)$ then $W$ is a gauge invariant function in $\mathcal{V}_\II(\p).$
\end{prop}
\begin{proof}
Let $W$ be an element in  $\WW(\g,f,k)$. Note that 
\begin{enumerate}[(i)]
\item in the proof of Proposition \ref{Prop:3.3_0207}, we showed that $W\in \WW(\g,f,k)$ if and only if $W([X, \mathcal{L}])=0$ for any $X\in (\n \otimes \mathcal{V}_\II(\p))_{\bar{0}}$, 
\item for $P\in \mathcal{V}_\II(\p)$ and a Lax operator $L= k\partial + \sum_{i\in I} q_i \otimes Q^i -f\otimes 1$, we have
\[ P([X,L])= P([X, \mathcal{L}])|_{q^i= Q^i}.\]
\end{enumerate}
Hence $W([X, L])=0$ for any $X\in (\n \otimes \mathcal{V}_\II(\p))_{\bar{0}}$ and any Lax operator $L$. Moreover, since $\ad^{n-1} X (L)$ is a Lax operator for $n\geq 1$, we have $W(\ad^n X (L))=0.$ Thus, 
\[ W(e^{\ad X}(L))=W(L),\]
which means that $W$ is gauge invariant.
\end{proof}

\begin{thm}\label{Thm:3.11_0730}
The set of gauge invariant functions in $\mathcal{V}_{\II}(\p)$ is the classical affine W-superalgebra associated to $\g$ and $f$.
\end{thm}

\begin{proof}
It directly follows from Proposition \ref{Prop:3.3_0207} and Proposition \ref{Prop:3.4_0211}.
\end{proof}

Now the following propositions are useful to find generators of W-superalgebras.

\begin{prop} \label{Prop:3.6_0212}
Let us fix an $\ad \, h$-invariant complementary subspace $V_f\subset \g$  of $[f, \n]$ in $\m^\perp$:  
\[ \m^{\perp}= V_f \oplus [f, \n].\]
Consider a Lax operator $L=k \partial + Q -f\otimes 1$ for $Q\in \mathcal{F}_{\g,f}$. Then there exists unique $Q^{can}\in V_f \otimes (\mathcal{V}_{\II}(\p))_{\bar{0}} $ and unique $X\in (\n \otimes \mathcal{V}_{\II}(\p))_{\bar{0}}$ such that 
\begin{equation} \label{Eqn:3.11_0212}
 e^{\ad X}L =k \partial +Q^{\text{can}}-f\otimes 1.
 \end{equation}
\end{prop}

\begin{proof}
 We can write $Q=\sum_{i\geq -\frac{1}{2}} Q_i$ where $Q_i\in \g(i) \otimes  \mathcal{V}_{\II}(\p) $. Similarly, let $X= \sum_{i\geq \frac{1}{2}}X_i$ and $Q^{\text{can}}= \sum_{i\geq 0} Q^{can}_i$ for $X_i, Q^{\text{can}}_i \in \g(i)\otimes  \mathcal{V}_{\II}(\p).$ Then the $\ad\, \frac{h}{2}$-decomposition of (\ref{Eqn:3.11_0212}) implies the following equalities:
\begin{equation}\label{Eqn:3.12_0212}
\begin{aligned}
&  Q_{-1/2}+[X_{1/2},-f\otimes 1] =0,\\
& Q_0 +[X_1, -f\otimes 1]+[X_{1/2}, Q_{-1/2}]= Q^{\text{can}}_0 ,\\
&  Q_{1/2} + [ X_{3/2}, -f\otimes 1] +[X_1, Q_{-1/2}]+[X_{1/2}, \partial+Q_0]=Q^{\text{can}}_{1/2}, \\
& \quad \vdots
\end{aligned}
\end{equation}
Then we can determine $X_{1/2}$ uniquely by the first equation in (\ref{Eqn:3.12_0212}) and it is even. Also, $X_1$ and $Q^{\text{can}}_0$ can be uniquely determined by $X_{1/2}$ and  the second  equation in (\ref{Eqn:3.12_0212}). Since $[X_{1/2}, Q_{-1/2}]+Q_0$ is even, both $X_1$ and $Q^{\text{can}}_0$ are even. The Inductively, even elements $X_{i+1}$ and $Q^{can}_{i}$ are determined uniquely by $X_{j+1}$, $Q^{\text{can}}_{j}$ for $j<i$ and  (\ref{Eqn:3.12_0212}).  
\end{proof}

\begin{lem} \label{Lem:label_gen}
Let us take the universal Lax operator $\mathcal{L}=k \partial +q_{\text{univ}} -1\otimes f$  in Proposition \ref{Prop:3.6_0212} and let $\{\, q_i\, |\,  i\in \mathcal{J}\, \}$ be a basis of $V_f$. If we denote $q^{\text{can}}_{\text{univ}}= \sum_{i\in \mathcal{J}} q_i \otimes w^i$ then we have the following properties. 
\begin{enumerate}
\item For any $q\in \mathcal{F}_{\g,f}$, we have 
\begin{equation}
\textstyle q^{\text{can}}=\sum_{i\in \mathcal{J}} q_i\otimes (w^i(q)).
\end{equation}
\item Let $P\in \mathcal{V}_{\II}(\p)$ and $q\in \mathcal{F}_{\g,f}$. Then we have  $P(q)=P(q^{\text{can}})$ if and only if $P(q_{\text{univ}})=P(q^{\text{can}}_{\text{univ}})$.
\end{enumerate}
\end{lem}
\begin{proof}
(1)  For any $q\in \mathcal{F}_{\g,f}$ and the basis $\{q^i\}_{i\in I}$ of $\p$ such that $(q_i|q^j)=\delta_{ij}$,   
the Lax operator \[ \textstyle L=k\partial + q-f\otimes 1=k\partial+\sum_{i\in I} q_i\otimes (q^i(q))-f\otimes 1.\] 
Also, for $X=n\otimes r\in \n \otimes \mathcal{V}_{\II}(\p)$ such that $e^{\ad X}(\mathcal{L})=k \partial+q^{\text{can}}_{\text{univ}}-f\otimes 1$, if we let $X_q:= n\otimes (r(q))$ then $e^{\ad X_q}(L)=k\partial+q^{X_q}-1\otimes f$ is obtained from $e^{\ad X}(\mathcal{L})$ by substituting $q^i$ in $q_{\text{univ}}^{\text{can}}$by $q^i(q).$ 
In other words, $q^{X_q} \in V_f \otimes \mathcal{V}_{\II}(\p)$ and 
\[ \textstyle q^{X_q}=q^{\text{can}}= \sum_{i\in \mathcal{J}} q_i\otimes (w^i(q)).\]

(2) It is enough to show that $P(q_{\text{univ}})=P(q^{\text{can}}_{\text{univ}})$ implies $P(q)=P(q^{\text{can}})$ for any $q\in \mathcal{F}_{\g,f}.$ Suppose $P(q_{\text{univ}})=P(q^{\text{can}}_{\text{univ}})$. 
Then, by (1), we have 
 \[ \textstyle P(q^{can})=P\big(\sum_{i\in \mathcal{J}} q_i \otimes w^i\big)|_{q^i=Q^i}=P(q_{\text{univ}}^{\text{can}})|_{q^i=Q^i}=P(q_{\text{univ}})|_{q^i=Q^i}= P(q)\]
where $|_{q^i=Q^i}$ denotes we substitute $q^i$ by $Q^i.$
\end{proof}

\begin{prop}\label{Prop:3.14_0730}
For $q^{\text{can}}_{\text{univ}}= \sum_{i\in \mathcal{J}} q_i \otimes w^i\in V_f \otimes \mathcal{V}_{\II}(\p)$ in Lemma \ref{Lem:label_gen}, we have
 \begin{equation}
 \WW(\g,f,k)=\Cdiff[\, w^i\, |\, i\in \mathcal{J}\, ] \subset \mathcal{V}_\II(\p) 
 \end{equation}
as differential superalgebras.
\end{prop}

\begin{proof}
  Note that  
   we have
\[\textstyle  \text{ (i) } w^i= w^i(q_{\text{univ}}) \text{ since } q_{\text{univ}}=\sum_{i\in I} q_i\otimes q^i, \qquad \text{ (ii) }  w^i=q^i(q^{\text{can}}_{\text{univ}}). \]
Consider the subset  $\{\, q^i\, |\,  i\in \mathcal{J}\, \}$ of the basis $\{q^i\}_{i\in I}$ of $\p$
and take an element $\Phi\in \WW(\g,f,k)$. Then 
\[  \Phi=\Phi(q_{\text{univ}})= \Phi( q^{\text{can}}_{\text{univ}}) \in  \Cdiff[\, q^i(q^{\text{can}}_{\text{univ}})\,  |\, i\in \mathcal{J}\, ]=\Cdiff[\, w_i \, |\, i\in \mathcal{J}\, ]\]
 and 
$\Phi\in \Cdiff[\, w^i\, |\, i\in \mathcal{J}\, ] .$ Hence, by Lemma \ref{Lem:label_gen}, we have $\WW(\g,f,k)\subset \Cdiff[\, w_i\, |\, i\in \mathcal{J}\, ] .$

Conversly, for $q, q'\in \mathcal{F}_{\g,f}$ such that $q\sim q'$, we have $q^{\text{can}}= q'^{\text{can}}$. Since $w^i(q)= q^i(q^{\text{can}})=  q^i(q'^{\text{can}})=w^i(q'^{\text{can}})$, we get $ \Cdiff[\, w^i\, |\, i\in \mathcal{J}\, ] \subset \WW(\g,f,k).$
\end{proof}

By Proposition \ref{Prop:3.14_0730}, we can find generators of $\WW(\text{spo}(2|1), f,k)$ as follows (cf. Example \ref{Ex:W(spo(2,1))}).

\begin{ex} \label{Ex:spo(1|2)W}
Let $\g=\text{spo}(2|1)$ and let 
\[ h= e_{11}-e_{22}, \quad e=e_{12}, \quad f=e_{21}, \quad e_{od}=e_{1\bar{1}}+e_{\bar{1}2}, \quad f_{od}=e_{2\bar{1}}-e_{\bar{1}1}\]
where $e_{ij}$ is the matrix in $\mathfrak{gl}(2|1)$ which has $1$ in the $ij$-entry and $0$ in other entries. Then $h,e,f$ are even elements and $e_{od}, f_{od}$ are odd elements. Lie brackets between generators are
\[ \, [h, e_{od}]= e_{od}, \ [h, f_{od}]=-f_{od}, \ [e_{od}, f_{od}]= [f_{od}, e_{od}]= -h, \ [e_{od}, f]=-f_{od},\]
\[ \, [f_{od}, e]=-e_{od}, \ [e_{od}, e_{od}]= 2e, \ [f_{od}, f_{od}]= -2f.\]
The even supersymmetric bilinear form we consider satisfies
\[ (h|h)= 2(e|f)=2, \ (e_{od}|f_{od})=-2.\]
Take the Lax operator 
\[ \mathcal{L}=k \partial +\frac{1}{2} f_{od}\otimes e_{od} +\frac{1}{2} h\otimes h -\frac{1}{2} e_{od}\otimes f_{od} + e\otimes f -f \otimes 1\]
and $V_f= \CC e \oplus \CC e_{od}$. We can check that $\m^{\perp}= V_f\oplus [\n,f].$ 
Suppose  $X= e_{od}\otimes r_{od}+e\otimes r$ for $r_{od} \in \mathcal{V}_{\II}(\p)_{\bar{1}}$ and $r\in \mathcal{V}_{\II}(\p)_{\bar{0}}$ satisfies 
\begin{equation}\label{Eqn:exLcan}
 e^{\ad X} \mathcal{L}= \mathcal{L}^{can}\in \CC \partial \ltimes V_f \otimes  \mathcal{V}_{\II}(\p).
 \end{equation}
The terms with $\ad h$-eigenvalue $-1$  in (\ref{Eqn:exLcan})  are
\[\frac{1}{2} f_{od} \otimes e_{od} - [e_{od},f] \otimes r_{od} =0.\]
Hence $r_{od}=-\frac{1}{2} e_{od}.$ The terms with $\ad h$-eigenvalue $0$ in (\ref{Eqn:exLcan})  are
\[\frac{1}{2} h\otimes h+ h\otimes ( -\frac{1}{2} e_{od}r_{od}-r) +\frac{1}{2} [e_{od}, f_{od}]\otimes  r_{od}^2 =0.\]
Since $e_{od}r_{od}=r_{od}^2=0$, we have $r=\frac{1}{2}h.$ Hence $X=-\frac{1}{2} e_{od}\otimes e_{od}+\frac{1}{2} e\otimes h$ and, by direct computations,
\begin{equation}
\begin{aligned}
 e^{\ad X} \mathcal{L}& =k \partial+e_{od} \otimes \big( -\frac{1}{2} f_{od}+ \frac{k}{2} \partial e_{od} +\frac{1}{4} h e_{od}\big) \\
 & + e\otimes \big( f+\frac{1}{2} f_{od}e_{od}-\frac{1}{4} h^2+\frac{k}{4} e_{od}\partial e_{od} -\frac{k}{2}\partial h \big) -f\otimes 1.
 \end{aligned}
 \end{equation}
 If we denote $\phi_{od}=  -\frac{1}{2} f_{od}+ \frac{k}{2} \partial e_{od} +\frac{1}{4} h e_{od}$ and $\phi_{ev}=  f+\frac{1}{2} f_{od}e_{od}-\frac{1}{4} h^2+\frac{k}{4} e_{od}\partial e_{od} -\frac{k}{2}\partial h$ then  $\mathcal{W}(\g,f,k)= \Cdiff[\phi_{od}, \phi_{ev}].$
\end{ex}

\begin{ex} \label{Ex:W(sl(2|1))}
Let $\g= \sll(2|1)$ and take $e= e_{12}$ and $f=e_{21}$. Consider the Lax operator 
\[ \mathcal{L}=k\partial +Q-f\otimes 1 = k\partial+e\otimes f+e_{1\bar{1}}\otimes e_{\bar{1}1}-e_{\bar{1}2}\otimes e_{2\bar{1}}+\frac{1}{2}h\otimes h -\frac{1}{2} \tau\otimes \tau -f \otimes 1,\]
where $h=[e,f]= e_{11}-e_{22}$ and $\tau= e_{11}+e_{22}+2e_{\bar{1}\bar{1}}\in \ker{\ad f}.$ Fix 
\[ V_f= \CC e_{12} \oplus \CC e_{1\bar{1}} \oplus \CC e_{\bar{1}2} \oplus \CC \tau\]
so that $[\n, f ]\oplus V_f = \m^\perp.$ Let us take
\[ X=e_{1\bar{1}}\otimes X_{1\bar{1}} + e_{\bar{1}2} \otimes X_{\bar{1}2} + e\otimes X_{12}\]
such that $e^{\ad X} \mathcal{L} \in \partial \ltimes V_f \otimes \mathcal{V}_{\II}(\p).$ In order to vanish degree $\frac{1}{2}$-part of $e^{\ad X} \mathcal{L}$, which is 
\[ -e_{\bar{1}1}\otimes e_{1\bar{1}} +e_{2\bar{1}}\otimes e_{\bar{1}2} + e_{2\bar{1}} \otimes X_{1\bar{1}}-e_{\bar{1}1}\otimes X_{\bar{1}2},\]
set $X_{1\bar{1}}= -e_{\bar{1}2}$ and $X_{\bar{1}2}= -e_{1\bar{1}}.$ Moreover, we want degree $0$-part of $e^{\ad X} \mathcal{L}$ lies in $\partial \ltimes \CC \tau \otimes \mathcal{V}_{\II}(\p)$. In other words,
\begin{equation*}
\begin{aligned}
& \mathcal{L}_0+[ X_{1/2}, \mathcal{L}_{-1/2}] + [X_1, L_{-1}]+\frac{1}{2} [X_{1/2}, [X_{1/2}, f]]\\
& = k\partial +\frac{1}{2}h\otimes h -\frac{1}{2} \tau\otimes \tau -\frac{1}{2} (h+\tau) \otimes e_{1\bar{1}}X_{1\bar{1}}-\frac{1}{2}(h-\tau)\otimes e_{\bar{1}{2}}X_{\bar{1}{2}} \\
& - h \otimes X_{12} -\frac{1}{2} [e_{\bar{1}2}, [ e_{1\bar{1}}, f]]\otimes X_{1\bar{1}}X_{\bar{1}2}-\frac{1}{2} [e_{1\bar{1}}, [e_{\bar{1}2}, f]]\otimes X_{\bar{1}2}X_{1\bar{1}}\\
& = k\partial + h\otimes (\frac{1}{2} h-X_{12})+\tau\otimes (-\frac{1}{2} \tau+ \frac{1}{2} e_{1\bar{1}}e_{\bar{1}2})\in \partial \ltimes \CC \tau \otimes \mathcal{V}_{\II}(\p).
\end{aligned}
\end{equation*}
Here, $\mathcal{L}_l$ and $X_l$ denote the degree $l$-parts of $\mathcal{L}$ and $X$, respectively. Hence, we fix $X_{12}= \frac{1}{2} h.$ In other words, we have 
\[ X= e_{12}\otimes \frac{1}{2} h - e_{\bar{1}2}\otimes e_{1\bar{1}}-e_{1\bar{1}}\otimes e_{\bar{1}2}.\]
Now, we compute degree $1/2$-part and $1$-part of $e^{\ad X} \mathcal{L}$ as follows. For degree $1/2$-part, we have
\begin{equation*}
\begin{aligned}
& \mathcal{L}_{1/2}= e_{1\bar{1}}\otimes e_{\bar{1}1}-e_{\bar{1}2}\otimes e_{2\bar{1}},\\
& [X_{1/2}, \mathcal{L}_0]= e_{\bar{1}2} \otimes ( k\partial e_{1\bar{1}}+\frac{1}{2} e_{1\bar{1}}(h-\tau)) + e_{1\bar{1}}\otimes (k\partial e_{\bar{1}2}+\frac{1}{2} e_{\bar{1}2}(h+\tau)),\\
& [X_1, \mathcal{L}_{-1/2}] = e_{\bar{1}2} \otimes \frac{1}{2} e_{1\bar{1}} h + e_{1\bar{1}} \otimes \frac{1}{2}e_{\bar{1}2} h,\\
& \frac{1}{2} [X_{1/2}, [X_{1/2}, \mathcal{L}_{-1/2}]] + \frac{1}{2} [X_{1/2}, [X_{1}, \mathcal{L}_{-1}] ]= -\frac{1}{4} e_{\bar{1}2}\otimes h e_{1\bar{1}}-\frac{1}{4} e_{1\bar{1}}\otimes he_{\bar{1}2},\\
& \frac{1}{2} [X_{1}, [X_{1/2}, \mathcal{L}_{-1}]] = -\frac{1}{4} e_{\bar{1}2}\otimes h e_{1\bar{1}}-\frac{1}{4} ,e_{1\bar{1}}\otimes he_{\bar{1}2}.
\end{aligned}
\end{equation*}
By adding all the formulas, we get $1/2$-degree part of $e^{\ad X} \mathcal{L}$: 
\[ e_{1\bar{1}} \otimes (e_{\bar{1}1} +k\partial e_{\bar{1}2} +\frac{1}{2} e_{\bar{1}2} (h+\tau))
+ e_{\bar{1}2}\otimes (-e_{2\bar{1}}+k\partial e_{1\bar{1}}+\frac{1}{2} e_{1\bar{1}}(h-\tau)).\]
Similarly, degree $1$-part of $e^{\ad X} \mathcal{L}$ is obtained by adding the following fomulas:
\begin{equation*}
\begin{aligned}
& \mathcal{L}_1= e\otimes f ,\\
& [X_{1/2}, \mathcal{L}_{1/2}]= e \otimes (-e_{1\bar{1}}e_{1\bar{1}} +e_{2\bar{1}}e_{\bar{1}2}) ,\\
& [X_1, \mathcal{L}_0]= e\otimes ( -\frac{k}{2} \partial h -\frac{1}{2} h^2),\\
& \frac{1}{2} ( [X_{1/2},[X_{1/2}, \mathcal{L}_0]]+[X_{1/2}, [X_1, \mathcal{L}_{-1/2}]])=e \otimes \frac{1}{2} ( -k \partial e_{1\bar{1}} e_{\bar{1}2} - k \partial e_{\bar{1}2} e_{1\bar{1}} + \tau e_{1\bar{1}}e_{\bar{1}2}), \\
& \frac{1}{2} ( [X_{1},[X_{1/2}, \mathcal{L}_{-1/2}]]+[X_{1}, [X_1, \mathcal{L}_{-1}]])= e\otimes \frac{1}{4} h^2,\\
& \frac{1}{6} ([X_{1/2}, [X_{1/2},[X_{1/2}, \mathcal{L}_{-1/2}]]] +[X_{1/2}, [X_{1/2}, [X_{1},\mathcal{L}_{-1}]]]\\
& \qquad \qquad+[X_{1/2}, [X_1, [X_{1/2},\mathcal{L}_{-1}]]]+ [X_{1},[X_{1/2}, [X_{1/2}, \mathcal{L}_{-1}]]])=0.
\end{aligned}
\end{equation*}
Hence the degree $1$ part of $e^{\ad X}\mathcal{L}$ is 
\[ e\otimes ( f-e_{\bar{1}1} e_{1\bar{1}}+e_{2\bar{1}}e_{\bar{1}2}-\frac{1}{4} h^2-\frac{k}{2} \partial h-\frac{k}{2} \partial e_{1\bar{1}}e_{\bar{1}2}-\frac{k}{2} \partial e_{\bar{1}2}e_{1\bar{1}} + \frac{1}{2} \tau e_{1\bar{1}}e_{\bar{1}2}).\]
If we denote 
\begin{equation}
\begin{aligned}
& \phi_t= -\frac{1}{2} \tau+\frac{1}{2} e_{1\bar{1}}e_{\bar{1}2},\\
& \phi_{1\bar{1}}= e_{\bar{1}1}+k\partial e_{\bar{1}2} +\frac{1}{2} e_{\bar{1}2}(h+\tau),\\
& \phi_{\bar{1}2}= -e_{2\bar{1}}+k\partial e_{1\bar{1}}+\frac{1}{2} e_{1\bar{1}}(h-\tau),\\
&\phi_e= f-e_{\bar{1}1} e_{1\bar{1}}+e_{2\bar{1}}e_{\bar{1}2}-\frac{1}{4} h^2-\frac{k}{2} \partial h-\frac{k}{2}\partial e_{1\bar{1}}e_{\bar{1}2}-\frac{k}{2} \partial e_{\bar{1}2}e_{1\bar{1}} + \frac{1}{2} \tau e_{1\bar{1}}e_{\bar{1}2}
\end{aligned}
\end{equation}
then \[\WW(\g,f, k)= \CC_{\text{diff}}[\phi_t, \phi_{1\bar{1}}, \phi_{\bar{1}2}, \phi_e].\]
Here, we note that 
\[ \phi_e= -k L-2\phi_\tau^2,\]
for the energy momentum field $L$, which is defined in Proposition \ref{Prop:2.11_0731}.
\end{ex}


Now, we introduce another equivalent way to get a generating set of a W-superalgebra, which will be used in the later sections to describe bi-Poisson structures via Lax operators. We consider the {\it sign twisting} linear map
\begin{equation}\label{signmap}
 \sigma: \g \otimes \mathcal{V}_\II(\p) \to  \g\otimes  \mathcal{V}_\II(\p), \quad a\otimes F\mapsto (a\otimes F)^\sigma,
 \end{equation}
where $(a\otimes F)^\sigma :=(-1)^{\pa(a)\pa(F)} a\otimes F.$

\begin{rem}
One can check that the map (\ref{signmap}) induces a Lie algebra automorphism on $(\g\otimes \mathcal{V}_\II(\p))_{\bar{0}}$. Moreover, if we consider the Lie algebra 
$(\mathcal{V}_\II(\p)\otimes \g)_{\bar{0}}$ with the bracket $[F\otimes a, G\otimes b]= (-1)^{\pa(a)\pa(b)}FG\otimes[a,b]$ then 
there is a natural isomorphism
\[ \phi: (\g\otimes  \mathcal{V}_\II(\p))_{\bar{0}} \to (\mathcal{V}_\II(\p)\otimes \g)_{\bar{0}}\]
defined by  $\phi((a\otimes F)^\sigma) = F\otimes a.$ In other words,  applying the map $\sigma$ to $ (\g\otimes  \mathcal{V}_\II(\p))_{\bar{0}}$ is equivalent to  consider the space $(\mathcal{V}_\II(\p)\otimes \g)_{\bar{0}}$.
\end{rem}

\begin{prop} \label{Prop:signLax}
For a Lax operator $L=k \partial + \sum_{i\in I} q_i\otimes Q^i -f \otimes 1 $, we denote 
\[  \textstyle L^\sigma=  k\partial + \sum_{i\in I} (-1)^{\pa(i)} q_i\otimes Q^i -f \otimes 1 \in (\g[z] \otimes \mathcal{V}_\II(\p))_{\bar{0}},\]
where $\pa(i)$ denotes the parity $\pa(q_i)$.
\begin{enumerate}
\item  Two elements $q, q'\in \mathcal{F}_{\g,f}$ are gauge equivalent if and only if $q^\sigma, q'^\sigma$ are gauge equivalent.
\item Recall that we denote the universal Lax operator by  $\mathcal{L}=k\partial +q_{\text{univ}}- f\otimes 1$ and the operator gives rise to generators of $\WW(\g,f,k)$ (see Proposition  \ref{Prop:3.14_0730}). By (1), we can find  generators of $\WW(\g,f,k)$
using {\it sign twisted universal Lax operator}
\[ \textstyle  \mathcal{L}^\sigma= k\partial +q_{\text{univ}}^\sigma-f\otimes 1.\]
\end{enumerate}
\end{prop}

\begin{proof}
The proof follows from the fact that 
\[  e^{\ad(n\otimes r)^\sigma}(L^\sigma)= L'^\sigma \quad \text{ if and only if } \quad e^{\ad (n\otimes r)}(L)= L'.\]
For (2), we note that 
\[ \textstyle q_{\text{univ}}^{\text{can}}= \sum_{i\in \mathcal{J}} q_i \otimes w^i \text{ if and only if } (q^\sigma_{\text{univ}})^{\text{can}}= \sum_{i\in \mathcal{J}} \pa(i)q_i \otimes w^i.\]
Since $\Cdiff[w^i| i \in \mathcal{J}]= \Cdiff[\pa(i) w^i|i \in \mathcal{J}]$, we proved (2).
\end{proof}


\section{Lax operators and Lie brackets on $\WW(\g,f,k)/\partial\WW(\g,f,k)$ }

\subsection{Derivatives on a differential superalgebra} \ 

Fix the  superalgebra of differential polynomials 
\[\, \PP=\Cdiff[\, u_i\, |\, i\in I\, ].\, \] 
The index set $I =I_{\bar{0}}\cup I_{\bar{1}}$ consists of two subindex sets $I_{\bar{0}}$ and $I_{\bar{1}}$ such that $u_i$ is even if $i\in I_{\bar{0}}$ and is odd if $i\in I_{\bar{1}}.$

\begin{defn} \label{Def:deriv}
Take $j_1, j_2, \cdots, j_k\in I_{\bar{1}}$ and $n_1, n_2, \cdots, n_k\in \ZZ_{\geq 0}$ such that  $(j_{k_1}, n_{k_1})\neq (j_{k_2},n_{k_2})$ for distinct numbers $k_1, k_2 \in \{1, \cdots, k\}$. Consider an element 
\[ \textstyle \phi_1=  u_{j_1}^{(n_1)} u_{j_2}^{(n_{2})} \cdots u_{j_k}^{(n_k)} \in \Cdiff[u_i | i\in I_{\bar{1}}],\]
where $u_j^{(n)}= \partial^n u_{j}$ and a monomial 
\[ \phi= \phi_0\phi_1 \in \PP, \text{ for } \phi_0\in  \Cdiff[u_i | i\in I_{\bar{0}}].\] 
\begin{enumerate}
\item The {\it (left) derivative} of $\phi$ with respect to $u_{j_t}^{(n_t)}$ for $t=1, \cdots, k$ is 
\[ \textstyle  \frac{\partial \phi}{\partial u_{j_t}^{(n_t)}} = (-1)^{t-1}\,  \phi_0 \cdot u_{j_1}^{(i_1)}u_{j_{2}}^{(i_2)} \cdots u_{j_{t-1}}^{(n_{t-1})}  u_{j_{t+1}}^{(n_{t+1})}  \cdots  u_{j_{k-1}}^{(n_{k-1})} u_{j_{k}}^{(n_{k})}.  \]
If $u_j^{(n)}\neq  u_{j_t}^{(n_t)}$ is an odd element then  $ \frac{\partial \phi}{\partial u_{j}^{(n)}}=0.$ If $j\in I_{\bar{0}}$, we let 
$  \frac{\partial \phi}{\partial u_{j}^{(n)}}= \frac{\partial \phi_0}{\partial u_{j}^{(n)}}\cdot \phi_1.$
\item The {\it right  derivative} of $\phi$ with respect to $u_{j_t}^{(n_t)}$ for $t=1, \cdots, k$ is 
\[ \textstyle  \frac{\partial_R \phi}{\partial_R u_{j_t}^{(n_t)}} = (-1)^{k-t}\,  \phi_0 \cdot u_{j_1}^{(i_1)}u_{j_{2}}^{(i_2)} \cdots u_{j_{t-1}}^{(n_{t-1})}  u_{j_{t+1}}^{(n_{t+1})}  \cdots  u_{j_{k-1}}^{(n_{k-1})} u_{j_{k}}^{(n_{k})}.  \]
If $u_j^{(n)}\neq  u_{j_t}^{(n_t)}$ is an odd element then  $ \frac{\partial_R \phi}{\partial_R u_{j}^{(n)}}=0.$ If $j\in I_{\bar{0}}$, we let 
$  \frac{\partial_R \phi}{\partial_R u_{j}^{(n)}}= \frac{\partial \phi}{\partial u_{j}^{(n)}}$.

\item The {\it (left) variational derivative} and  {\it right variational derivative} of $\phi$ with respect to $u_i$ are
\begin{equation}
 \textstyle \frac{\delta \phi}{\delta u_i} = \sum_{n\in \ZZ_{\geq 0}} (-\partial)^n \frac{\partial \phi}{\partial u_i^{(n)}}\ , \qquad  \frac{\delta_R \phi}{\delta_R u_i}= \sum_{n\in \ZZ_{\geq 0}} (-\partial)^n \frac{\partial_R \phi}{\partial_R u_i^{(n)}}.
\end{equation}

\item   The {\it  variational derivative} of $\phi$ with respect to $\{ \, u_i \, |\, i\in I\, \}$ is 
\[ \textstyle \frac{\delta \phi}{\delta u}= \sum_{i\in I} u_i \otimes \frac{\delta \phi}{\delta u_i}.\]
\end{enumerate}
The derivatives defined in (1), (2), (3), (4) have linearities so that they are well-defined on the set of differential algebra $\PP=\Cdiff[u_i|i\in I].$
\end{defn}

\begin{rem} 
  We have 
\[ \textstyle \frac{\delta \phi}{\delta u_i} = (-1)^{\pa(u_i)\pa(\phi\cdot u_i)} \frac{\delta_R \phi}{\delta_R u_i}. \]
Hence, for the map $\sigma$ in (\ref{signmap}), 
\[ \textstyle \left(\frac{\delta \phi}{\delta u}\right)^\sigma=  \sum_{i\in I}  \left(u_i \otimes \frac{\delta \phi}{\delta u_i }\right)^\sigma= \sum_{i\in I} u_i \otimes \frac{\delta_R \phi}{\delta_R u_i } .\]
\end{rem}

\begin{rem}
If $u_i\in \PP$  is an odd element then  $\frac{\partial}{\partial u_i^{(n)}}$ for any $n\in \ZZ_{\geq 0}$ is an odd derivation, i.e. 
\[\textstyle \frac{\partial}{\partial u_i^{(n)}} (FG)= \left(\frac{\partial}{\partial u_i^{(n)}} F \right)\cdot G + (-1)^{\pa(F)} F\cdot \left(\frac{\partial}{\partial u_i^{(n)}}G\right).\]
\end{rem}

\begin{prop} \label{Prop:master formula}
 For homogeneous elements $f,g \in \PP$, we have 
 \begin{equation} \label{Eqn:master}
\textstyle \{ f_\lambda g\} = \sum_{\substack{i,j\in I \\ m,n\in \ZZ_{\geq 0} }}
 C_{i,j}^{f,g} \frac{\partial_R g}{\partial_R u_j^{(n)}}(\lambda+\partial)^n \{u_{i\, \lambda+\partial\, } u_j\}_{\to}(-\lambda-\partial)^m \frac{\partial f}{\partial u_i^{(m)}}
 \end{equation}
where $ C_{i,j}^{f,g}=(-1)^{\pa(f)\pa(g)+\pa(i)\pa(j)}.$  Note that 
\[ \textstyle \{a_{\lambda+\partial} b\}_{\to} c:= \sum_{i\in \ZZ_{\geq 0}} \frac{1}{n!} a_{(n)}b (\lambda+\partial)^n c \quad \text{ for } \quad  a,b,c\in \mathcal{P}.\]
\end{prop}

\begin{proof}
The formula (\ref{Eqn:master}) follows from sesquilinearities and Leibniz rules of $\lambda$-brackets. The only part we have to be careful is the constant factor $C_{i,j}^{f,g}$ in (\ref{Eqn:master}).
One can see that \[ C_{i,j}^{f,g}= (-1)^{(\pa(i)+\pa(f))(\pa(j)+\pa(g))}\cdot C_{i,j}^{f}\cdot C_{j,i}^{g}\]
 for $C_{i,j}^{f}=(-1)^{(\pa(f)+\pa(i))\pa(j)}$ and $C_{j,i}^{g}=(-1)^{(\pa(g)+\pa(j))\pa(i)}.$ 
Note that switching the position of $\frac{\partial f}{\partial u_i^{(m)}}$ and $\frac{\partial_R g}{\partial_R u_j^{(n)}}$ gives rise to the $(-1)^{(\pa(i)+\pa(f))(\pa(j)+\pa(g))}$ in $C_{i,j}^{f,g}$ and switching the position of $u_j$ and $\frac{\partial f}{\partial u_i^{(m)}}$ (resp. $u_i$ and $\frac{\partial_R g}{\partial_R u_j^{(n)}}$) gives rise to the constant factor $C_{i,j}^{f}$ (resp. $C_{j,i}^{g}$). 
\end{proof}

\begin{prop} \label{Prop:vari_der}
\begin{enumerate}
\item For any variable $u_i^{(n)}$ in $\PP$, we have 
\begin{equation} \label{Eqn:4.2_1011}
 \textstyle \left[ \frac{\partial}{\partial u_i^{(n)}}, \partial \right] = \frac{\partial}{\partial u_i^{(n-1)}}.
 \end{equation}
\item Let $\phi\in \PP$. Then
$ \frac{\delta}{\delta u_i} \partial \phi =0$
for any $i\in I$.
\end{enumerate}
\end{prop}

\begin{proof}
(1) For any $j\in I$ and $m \in \ZZ_{\geq 0}$, we can check that $ \textstyle \left[ \frac{\partial}{\partial u_i^{(n)}}, \partial \right] (u_j^{(m)}) = \frac{\partial}{\partial u_i^{(n-1)}} (u_j^{(m)}).$
Suppose we get the same element in $\PP$ when we apply $F$ (resp. $G$) to the LHS and RHS of the equation (\ref{Eqn:4.2_1011}). If $i\in I_{\bar{1}}$ then
\begin{equation*}
\begin{aligned}
\textstyle \left[ \frac{\partial}{\partial u_i^{(n)}}, \partial \right] (FG)& \textstyle = \left(\frac{\partial}{\partial u_i^{(n)}}\partial F\right) G+(-1)^{\pa(F)}\partial F\frac{\partial}{\partial u_i^{(n)}}G+ \frac{\partial}{\partial u_i^{(n)}}F \partial G+(-1)^{\pa(F)}F\frac{\partial}{\partial u_i^{(n)}}\partial G  \\
 & \textstyle - \left( \partial\frac{\partial}{\partial u_i^{(n)}} F\right) G -(-1)^{\pa(F)}\partial F\frac{\partial}{\partial u_i^{(n)}}G- \frac{\partial}{\partial u_i^{(n)}}F \partial G-(-1)^{\pa(F)}F\partial \frac{\partial}{\partial u_i^{(n)}} G \\
 &\textstyle= \left[ \frac{\partial}{\partial u_i^{(n)}}, \partial \right](F) \cdot  G + (-1)^{\pa(F)}F \cdot \left[ \frac{\partial}{\partial u_i^{(n)}}, \partial \right](G) \\
 &\textstyle = \frac{\partial}{\partial u_i^{(n-1)}}F \cdot G + (-1)^{\pa(F)}F\frac{\partial}{\partial u_i^{(n-1)}}(G)= \frac{\partial}{\partial u_i^{(n-1)}}(FG).
  \end{aligned}
 \end{equation*}
For $i\in I_{\bar{0}}$, the same argument works. By induction, we proved (1).\\

(2) By (1), we have 
\[ \textstyle \sum_{i\in \ZZ_{\geq 0}}(-\partial)^n \frac{\partial}{\partial u_i^{(n)}}\partial=-\sum_{i\in \ZZ_{\geq 0}}(-\partial)^{n+1}\frac{\partial}{\partial u_i^{(n)}}  +\sum_{i\in \ZZ_{\geq 1}} (-\partial)^{n} \frac{\partial}{\partial u_i^{(n-1)}}=0.\]
\end{proof}

\subsection{Lie brackets on $\WW(\g,f,k)/\partial \WW(\g,f,k)$} \label{Subsec:Lie bracket}. \ 

Now we are ready to find Lie superalgebra structures on a quotient space of W-algebra via Lax operators. In this section, we assume $\g$ is a finite dimensional simple Lie superalgebra with even supersymmetric bilinear invariant form $(\, |\, )$. 

Let us denote  by $\g((z^{-1})):= \g[z]\oplus z^{-1}\g[[z^{-1}]]$  the Lie superalgebra endowed with the bracket $[a z^m, b z^n]= [a,b] z^{m+n}$ where $[a,b]$ is the Lie bracket on $\g$. The vector superspace   $\g((z^{-1}))\otimes \mathcal{V}_{\II}(\p)$ is a Lie superalgebra endowed with the bracket such that 
\[ \, [a z^m \otimes F, b z^n \otimes  G]= (-1)^{\pa(b)\pa(F)}[a z^m,bz^n]\otimes FG \]
for  $a, b\in \g$, $n,m\in \ZZ$ and $F,G \in \mathcal{V}_{\II}(\p).$ We extend the Lie bracket to that on $\CC\partial \ltimes  \g((z^{-1}))\otimes  \mathcal{V}_{\II}(\p)$  by considering  $\CC\partial$ as the trivial Lie algebra. \\

Consider the universal Lax operator associated to $\g$ and $\Lambda$: 
\begin{equation} \label{Lax_lambda}
 \mathcal{L}(\Lambda)= k\partial+q_{\text{univ}}-\Lambda\otimes 1 \in \CC\partial \ltimes \g((z^{-1}))\otimes \mathcal{V}_\II(\p),
 \end{equation}
where $q_{\text{univ}}$ is that in (\ref{Eqn:univ}) and $\Lambda=f+zs$ for a nonzero even element $s$  in $\g(d)$.  Here, $d$ is the largest eigenvalue such that $\g(d)\neq 0$. The operator acts on  $ \g((z^{-1}))\otimes \mathcal{V}_{\II}(\p)$ by the adjoint action. Note that we assume the even part of $\g(d)\neq 0$ (see Remark \ref{Rem:s}).


\begin{rem}
 For the operator $\mathcal{L}$ associated to $\g$ and $f$ in (\ref{Eqn:univ}), we have $\mathcal{L}(\Lambda)= \mathcal{L}-zs\otimes 1.$
Denote by
\begin{equation}
\mathcal{L}_{[1]}(\Lambda):=\mathcal{L} \quad \text{ and }\quad  \mathcal{L}_{[2]}(\Lambda):=-s \otimes 1.
\end{equation}
The term $\mathcal{L}_{[2]}(\Lambda)$ do the crucial role to define bi-Poisson structures on $\WW(\g,f,k)$. However, this part is not important when we find  generators of $\WW(\g,f,k)$.
\end{rem}

Also, using the map (\ref{signmap}), we denote the {\it sign twisted universal Lax operator} by
\begin{equation}\label{Lax_lambda_rev}
\mathcal{L}^\sigma(\Lambda)=k\partial + q^{\sigma}_{\text{univ}}-\Lambda\otimes 1 \in  \CC\partial \ltimes  \g((z^{-1})) \otimes \mathcal{V}_{\II}(\p).
\end{equation}

 Now, via $\mathcal{L}^\sigma(\Lambda)$,  we aim to define Lie brackets on $\WW(\g,f,k)/\partial \WW(\g,f,k)$ for the subspace $\partial \WW(\g,f,k):= \{ \, \partial W \, | \, W\in\WW(\g,f,k) \, \} \subset \WW(\g,f,k).$ (In Remark \ref{Rem:4.16_0721}, we explain why we consider $\mathcal{L}^\sigma(\Lambda)$ instead of $\mathcal{L}(\Lambda)$.)

 
 The following notion is analogous to the notion in (2) of Definition \ref{Def:HEviaPVA}.
 
 \begin{defn}
 Let $\mathcal{V}$ be a differential superalgebra with the derivation $\partial$ and denote $\partial \mathcal{V}:=\{ \, \partial V \, | \, V\in \mathcal{V}\, \}$. For the map 
 \[ \textstyle \int: \mathcal{V} \to \mathcal{V}/ \partial \mathcal{V}, \quad V\mapsto V+\partial \mathcal{V}=:\int V,\]
 we call $\int V$ the {\it local functional} of $V\in \mathcal{V}.$
 \end{defn}

Let the bilinear form $(\cdot|\cdot): \g((z^{-1})) \otimes \mathcal{V}_{\II}(\p) \times \g((z^{-1})) \otimes \mathcal{V}_{\II}(\p) \to \mathcal{V}_{\II}(\p)$ be defined by 
 \[ ( a z^m\otimes F  |  b z^n \otimes G )= (-1)^{\pa(b)\pa(F)}(a|b)\delta_{m+n, 0}FG\, .\]
 Define two bilinear brackets $\{ \, ,\, \}_i: \mathcal{V}_{\II}(\p)\times \mathcal{V}_{\II}(\p)\to  \mathcal{V}_{\II}(\p)$, for $i=1,2$,  by 
\begin{equation} \label{Poi_W}
\begin{aligned}
& \textstyle \{ \phi,\psi\}_1= \left( \frac{\delta \phi }{\delta q}\right. \left| \left[ \frac{\delta \psi }{\delta q}\,, \mathcal{L}^\sigma (\Lambda)\right]\right)=\left( \frac{\delta \phi }{\delta q}\right. \left| \left[ \frac{\delta \psi }{\delta q}\,, \mathcal{L}^\sigma_{[1]}(\Lambda)\right]\right),\\
& \textstyle \{\phi,\psi\}_2= - \left( \frac{\delta \phi }{\delta q}  \right. \left|  z^{-1} \left[ \frac{\delta \psi }{\delta q}\, , \mathcal{L}^\sigma(\Lambda)\right]\right)=-\left( \frac{\delta \phi }{\delta q}\right. \left|\left[ \frac{\delta \psi }{\delta q}\,, \mathcal{L}^\sigma_{[2]}(\Lambda)\right]\right),\\
\end{aligned}
\end{equation}
where $q=(q^i)_{i\in I}$ is the basis of $\p$ in (\ref{Eqn:univ}) and $z^{-1}(a z^n\otimes F):= a z^{n-1}\otimes F \in \g((z^{-1})) \otimes \mathcal{V}(\p)$. By Proposition \ref{Prop:vari_der} (2), we can consider the induced bilinear brackets on $\mathcal{V}_{\II}(\p)/\partial\mathcal{V}_{\II}(\p)$
\begin{equation}\label{Eqn:4.2_0413}
\textstyle \, [ \ , \  ]_i: \mathcal{V}_{\II}(\p)/\partial \mathcal{V}_{\II}(\p) \times \mathcal{V}_{\II}(\p)/\partial \mathcal{V}_{\II}(\p) \to  \mathcal{V}_{\II}(\p)/\partial \mathcal{V}_{\II}(\p), \quad i=1,2
\end{equation}
such that $[ \int \phi , \int \psi ]_i: = \int \{\phi, \psi\}_i.$



\begin{lem} \label{Lem:4.6_0411}
Let $ q' , r \in\mathcal{F}_{\g,f}= (\m^\perp\otimes\mathcal{V}_{\II}(\p))_{\bar{0}}$. For $\epsilon\in \CC$, we have 
\begin{equation}\label{Eqn:4.2}
 \left. \frac{\int \phi(q'+\epsilon r) - \int \phi(q')}{\epsilon}\right|_{\epsilon=0} = \int \left( r^\sigma \left| \frac{\delta \phi(q')}{\delta q}\right. \right).
\end{equation}
\end{lem}

\begin{proof}
If $r \in (\m^\perp)_{\bar{0}}\otimes (\mathcal{V}_{\II}(\p))_{\bar{0}}$ then the proof is similar to the non-super algebras cases in \cite{BDHM}. Suppose $r= q_i \otimes r^i$ for $q_i\in (\m^\perp)_{\bar{1}}.$ Then 
\begin{equation}\label{Eqn:4.9_0516}
\textstyle  \left( r^\sigma \left| \frac{\delta \phi(q')}{\delta q}\right. \right)=\left( -q_i\otimes r^i \left| \frac{\delta \phi(q')}{\delta q}\right. \right)= r^i (q_i\otimes 1 | q^i\otimes 1) \frac{\delta\phi(q')}{\delta q^i}= r^i \frac{\delta\phi(q')}{\delta q^i}. 
\end{equation}
Also, we can see the LHS of (\ref{Eqn:4.2}) is the same as $r^i \frac{\delta\phi(q')}{\delta q^i}.$ In detail, for $\phi= \partial^{n_1} (q^{i_1})\, \partial^{n_2} (q^{i_2})\,\cdots \, \partial^{n_t} (q^{i_t})$, if we denote the first $k-1$ terms in $\phi$ by $\phi_k=\partial^{n_1}( q^{i_1}) \, \partial^{n_2} (q^{i_2})\, \cdots \, \partial^{n_{k-1}} (q^{i_{k-1}})$ and the last $t-k$ terms  in $\phi$ by $\psi_k=\partial^{n_{k+1}} (q^{i_{k+1}}) \, \partial^{n_{k+2}} (q^{i_{k+2}})\, \cdots \, \partial^{n_{t}} (q^{i_{t}})$ for $k=0,1,\cdots, t$ then 
\[ \phi= \phi_k\cdot  \partial^{n_k} q^{i_k} \cdot \psi_k \]
for any $k.$ Now, for such $\phi$, let $r=q_i \otimes r^i$ and $q'=\sum_{j\in I} q_j \otimes Q^j$. Then we have
\begin{equation} \label{Eqn:4.10_0516}
\begin{aligned}
& \textstyle  \left.\frac{\phi(q'+\epsilon r)-\phi(q')}{\epsilon}\right|_{\epsilon=0}   = \sum_{k=0}^t \phi_k(q') \cdot \partial^{n_k} q^{i_k}(r) \cdot \psi_k(q') \\
& \textstyle = \sum_{k=1}^t (-1)^{\pa(q_i)\pa(\phi_k)}\partial^{n_k} q^{i_k}(r) \cdot \phi_k(q')\psi_k(q') \\
& \textstyle = \sum_{k=1}^t \partial^{n_k}\big( r^i (q_i | q^{i_k})\big)(-1)^{\pa(q_i)\pa(\phi_k)}   \phi_k(q')\psi_k(q') \\
& \textstyle = \sum_{n\geq 0 } \partial^n r^i \frac{\partial \phi}{\partial q^{i\, (n)}}.
\end{aligned}
\end{equation}
The last equality holds since $(q_i| q^j)= \delta_{ij}.$ Since $\int \partial(F) G= \int -F \partial(G)$, the lemma is proved by (\ref{Eqn:4.9_0516}) and  (\ref{Eqn:4.10_0516}).
\end{proof}

\begin{prop}
The brackets $[ \ , \ ]_i$ for $i=1,2$ satisfy skew-symmetry.
\end{prop}
\begin{proof}
For $\phi=a\otimes F$ and $ \psi=b\otimes G$ in $ \g \otimes \mathcal{V}_{\II}(\p)$, we have
\begin{equation*}
\begin{aligned}
& \textstyle \int (a\otimes F| [b\otimes G, \partial])=\int (a\otimes F |-b \otimes \partial G)= - \int (-1)^{\pa(\phi)\pa(\psi)} (b\otimes \partial G| a\otimes F) \\ & \textstyle = \int (-1)^{\pa(\phi)\pa(\psi)}(b\otimes G| a\otimes \partial F )= - \int (-1)^{\pa(\phi)\pa(\psi)}(b\otimes G| [a\otimes F, \partial]).
\end{aligned}
\end{equation*}
Hence, by invariance and skew-symmetry of the bilinear form on $\g \otimes \mathcal{V}_{\II}(\p),$ we have
\[\textstyle   [\int \phi, \int \psi]_i = \int \{\phi, \psi\}_i=- (-1)^{\pa(\phi)\pa(\psi)} \int \{\psi, \phi\}_i=- (-1)^{\pa(\phi)\pa(\psi)}[ \int \psi, \int \phi]_i , \quad i=1,2.\]
\end{proof}

\begin{lem} \label{lem:4.9_0413}
For $\psi \in \mathcal{V}_{\II}(\p)$, we have: 
\begin{equation*}
\begin{aligned}
 & \textstyle [ \int  q^i, \int \psi ]_1=\sum_{j\in I }\int\left([q^i, q^j]-  (q^i|q^j)k \partial \right)\frac{\delta}{\delta q^j} \psi, \\
& \textstyle [\int q^i,  \int \psi]_2=\sum_{j\in I }\int ([q^i,q^j]|s) \frac{\delta}{\delta q^j} \psi,\\
\end{aligned}
\end{equation*}
where $\{q^j|j\in I\}$ is a basis of $\p$.
\end{lem}
\begin{proof}
By expanding the LHS of $\int \{ q^i, \psi\}_1$, we obtain
\begin{equation}
\begin{aligned}
& \textstyle [\int q^i, \int \psi ]_1  =\textstyle  \sum_{j,j'\in I} \int \left( q^i\otimes 1 \right| \left. [ q^j\otimes \frac{\delta}{\delta q^j} \psi   , k \partial+\pa(j')q_{j'}\otimes q^{j'} -f\otimes 1]\right)\\
& = \textstyle  \sum_{j\in I}\int (q^i|q^j)(-k \partial)\frac{\delta}{\delta q^j} \psi+   \sum_{j, j'\in I}\int q^{j'}(q_{j'}|[q^i, q^j])  \frac{\delta}{\delta q^j} \psi-\sum_{j\in I} ([q^i,q^j]|f) \frac{\delta}{\delta q^j} \psi,
\end{aligned}
\end{equation}
where $\{q_{j}|j\in I\}$ is the basis  of $\m^\perp$ such that $(q_j| q^{j'})= \delta_{j,j'}$ and $\pa(j')$ is the parity of $q^{j'}$.
In $\mathcal{V}_\II(\p)$, we have 
\[ \textstyle \sum_{j\in I}[q^i, q^j]\frac{\delta}{\delta q^j} \psi=  \sum_{j, j'\in I}\int([q^i, q^j]|q_{j'})q^{j'}  \frac{\delta}{\delta q^j} \psi-\sum_{j\in I} ([q^i,q^j]|f) \frac{\delta}{\delta q^j} \psi.\]
Hence the first equality is proved. The second equality also can be proved similarly.
\end{proof}

\begin{prop} \label{Prop:Lie stru}
We have the following equations: 
\begin{equation}\label{Eqn:4.4_0413}
\begin{aligned}
 & \textstyle [ \int  \phi, \int \psi ]_1=\sum_{i,j\in I }\int \frac{\delta_R}{\delta_R q^i} \phi \left([q^i, q^j]-  (q^i|q^j)k\partial \right)\frac{\delta}{\delta q^j} \psi,\\
 &  \textstyle [\int \phi, \int \psi]_2=\sum_{i,j\in I }\int \frac{\delta_R}{\delta_R q^i} \phi  ([q^i,q^j]|s) \frac{\delta}{\delta q^j} \psi.
\end{aligned}
\end{equation}
\end{prop}

\begin{proof}
Observe that \[ \textstyle \left(\left. \frac{\delta}{\delta q} \phi \right| a\otimes F\right)= \sum_{i\in I} \frac{\delta_R \phi}{\delta_R q^i}( q^i\otimes 1|a\otimes F)\] for any $a \otimes F \in \g \otimes \mathcal{V}_{\II}(\p).$ Since $\frac{\delta}{\delta q} q^i= q^i\otimes 1$, we have 
\begin{equation}
\begin{aligned}
&  \textstyle [\int  \phi, \int  \psi]_t =\int \sum_{i\in I}  \frac{\delta_R \phi}{\delta_R q^i} \left(q^i\otimes 1 \left|  \left[ \frac{\delta}{\delta q}\psi , \mathcal{L}^{\sigma}_{[t]} (\Lambda) \right] \right. \right)  \\
& \textstyle = \int \sum_{i\in I}  \frac{\delta_R \phi}{\delta_R q^i} \left(\frac{\delta}{\delta q} q^i \left|  \left[ \frac{\delta}{\delta q}\psi , \mathcal{L}^{\sigma}_{[t]} (\Lambda) \right] \right. \right) = \int \sum_{i\in I} \frac{\delta_R \phi}{\delta_R q^i} \{ q^i, \psi\}_t
 \end{aligned}
 \end{equation}
 for $t=1,2.$ 
Now, by the proof of  Lemma \ref{lem:4.9_0413}, we can see (\ref{Eqn:4.4_0413}).
\end{proof}

By previous propositions, we know how to compute the brackets $[\, , \, ]_i$, $(i=1,2)$, defined on $\mathcal{V}_\II(\p) / \partial \mathcal{V}_{\II}(\p).$ The next thing we want to show is that the brackets can be understood as brackets  on $\WW(\g,f,k)/\partial \WW(\g,f,k)$:
\[ [\ , \ ]_i : \WW(\g,f,k)/\partial \WW(\g,f,k) \times \WW(\g,f,k)/\partial \WW(\g,f,k) \to \WW(\g,f,k)/\partial \WW(\g,f,k).\]
We provide proofs in two different ways. One (Proposition \ref{Prop:4.12_0804}) is by the definition W-algebras in Section \ref{Sec:W-super} and the other one (Proposition \ref{Prop:4.13_0804})  is purely algebraic. Note that the first proof is inspired from \cite{BDHM} and the second one is inspired form \cite{DKV}.

\begin{lem} \label{Lem:4.12_0516}
Let $\mathcal{V}$ be a superalgebra of differential  polynomials with both even generators and odd generators. 
 If $G,G'\in \mathcal{V}$ satisfy 
\[ \textstyle   \int FG= \int FG' \quad  \text{ or } \quad \overline{FG}= \overline{FG'} \text{ in } \mathcal{V}/\partial \mathcal{V} \]
for all $F\in\mathcal{V}_{\bar{0}}$ or all $F\in\mathcal{V}_{\bar{1}}$ then $G=G'.$
\end{lem}

\begin{proof}
Let us denote $\mathcal{V}= \mathcal{V}_{\bar{0}} \otimes \mathcal{V}_{\bar{1}}$, where $\mathcal{V}_{0}= \CC_{\text{diff}}[u_{i}|i=1,2, \cdots, {k_0}]$ and $\mathcal{V}_{1}= \CC_{\text{diff}}[v_{i}|i=1,2, \cdots, {k_1}]$
for even variables $u_i$ and odd variables $v_i.$ 

It is enough to show that if $\int  u_i^{(n)} G=0$ (resp. $\int  v_i^{(n)} G=0$) for any even  (resp. odd) variable $ u_i^{(n)}$ (resp. $ v_i^{(n)}$) then $G=0.$  

Let us first show that if  $\int  u_i^{(n)} G=0$ for any even  variable $ u_i^{(n)}$  then $G=0.$    If $G\in \CC^\times$ then $u_1G\not \in \partial\mathcal{V}$. Suppose $G$ is not a constant. Take an integer $m\in \ZZ_{\geq 0}$ such that no monomial in $G$ has terms with $u_1^{(n)}$ for $n\geq m.$ Then $u_1^{(m+1)} G \not \in \partial\mathcal{V}$. That is because if $\partial G_1= u_1^{(m+1)}G$ then $G_1$ should have the term $u_1^{(m+1)} G_2$ for a nonconstant element $G_2$. Hence $\partial G_1= u_1^{(m+1)}G$ has the term with $u_1^{(m+2)}$. This is a contradiction to our assumption and $G$ should be $0$.

Suppose  $\int v_i^{(n)} G=0$ for any odd variable $ v_i^{(n)}$ then $G=0.$    We have $G\not \in \CC^\times $ since otherwise $v_1G\not \in \partial\mathcal{V}$. Similarly to the previous case, for a nonconstant element $G$, let us take an integer $m\in \ZZ_{\geq 0}$ such that no monomial in $G$ has terms with $v_1^{(n)}$ for $n\geq m.$ Then $v_1^{(m+2)}G$ is not in $\partial \mathcal{V}.$ Here we note that $v_1^{(m+1)}G$ can be in $\partial \mathcal{V}$, for example $\partial(v_1^{(m)}v_1^{(m-1)})= v_1^{(m+1)}v_1^{(m-1)}$. Hence $G=0.$
\end{proof}

\begin{prop}\label{Prop:4.12_0804}
The brackets  (\ref{Poi_W}) induce brackets on $ \WW(\g,f,k)/\partial\WW(\g,f,k).$
\end{prop}

\begin{proof}
It is enough to show that for $\phi, \psi\in \WW(\g,f,k)$, we have $\{\phi, \psi\}_1$ and $\{\phi, \psi\}_2$ are also in $\WW(\g,f,k).$ 

Let $X\in (\n \otimes \mathcal{V}_{\II}(\p))_{\bar{0}}$ and two elements $q,q' \in \mathcal{F}_{\g,f}$ be gauge equivalent by $X$, i.e.
\[ k\partial+ q'^{\sigma}-f\otimes 1=e^{\ad X^{\sigma}}(k\partial+ q^{\sigma}-f\otimes 1).\]  
Note that, since $[s, \n]= 0$, we have 
\begin{equation} \label{Eqn:4.9_0805}
 k\partial+ q'^{\sigma}-\Lambda\otimes 1=e^{\ad X^{\sigma}}(k\partial+ q^{\sigma}- \Lambda\otimes 1).
 \end{equation}
If $r^{\sigma}:= e^{-\ad X^{\sigma}} r'^{\sigma}\in \mathcal{F}_{\g,f}$  for some $r' \in \mathcal{F}_{\g,f}$, we have 
\[ \phi(q)= \phi(q') \text{ and }  \phi(q'+\epsilon r')= \phi(q+\epsilon r) \text{ for } \epsilon\in \CC.\]
 Hence, by Lemma \ref{Lem:4.6_0411}, the equation
$ \phi(q+r)-\phi(q)= \phi(q'+r')-\phi(q')$ implies 
\begin{equation} \label{Eqn:4.16_0516}
\begin{aligned}
&  \textstyle \int \left(r'^{\sigma} \left| \frac{\delta \phi(q')}{\delta q}\right. \right)=  \int\left.\frac{\phi(q'+\epsilon r')-\phi(q')}{\epsilon}\right|_{\epsilon=0}  = \int\left.\frac{\phi(q+\epsilon r)-\phi(q)}{\epsilon}\right|_{\epsilon=0} \\
& \textstyle= \int \left(r^{\sigma} \left| \frac{\delta \phi(q)}{\delta q}\right. \right)  \textstyle = \int\left( e^{-\ad X^{\sigma}} r'^{\sigma} \left| \frac{\delta \phi(q)}{\delta q}\right. \right)= \int\left(r'^{\sigma} \left| e^{\ad X^{\sigma}}\frac{\delta \phi(q)}{\delta q}\right. \right).
\end{aligned}
\end{equation}
By Lemma \ref{Lem:4.12_0516} and (\ref{Eqn:4.16_0516}), we have $ \frac{\delta \phi(q')}{\delta q}=  e^{\ad X^{\sigma}}\frac{\delta \phi(q)}{\delta q}$. 
It is obvious that $\psi$ has the same property. 
Hence we have 
\begin{equation}
\begin{aligned}
&  \textstyle \{\phi, \psi\}_1(q')  = \left( \frac{\delta \phi(q')}{\delta q} \left| \left[\frac{\delta \psi(q')}{\delta q}, k\partial + q'^{\sigma}-1\otimes \Lambda\right]\right.\right) \\
 & \textstyle= \left( e^{\ad X^{\sigma}}\frac{\delta \phi(q)}{\delta q} \left| \left[e^{\ad X^{\sigma}}\frac{\delta \psi(q)}{\delta q}, e^{\ad X^{\sigma}} (k\partial + q^{\sigma}-1\otimes \Lambda) \right]\right.\right)= \{\phi, \psi\}_1(q).
 \end{aligned}
 \end{equation}
On the other hand, since $ e^{\ad X^{\sigma}} (s\otimes 1)= s\otimes 1$, we can prove that $\{\phi, \psi\}_2(q')= \{\phi,\psi\}_2(q)$ by the same arguement. 
\end{proof}

The classical W-superalgebra $\WW(\g,f,k)$ is a PVA endowed with the $\lambda$-brackets induced from those on the affine PVA $S(\CC[\partial]\otimes \g)$ such that 
\[ \, \{ a_\lambda b\}_1=[a,b]+k\lambda(a|b), \quad \{ a_\lambda b\}_2=(s|[a,b]) \quad\text{ for } a,b\in \g.\]
Hence there are Lie brackets $[\ , \ ]'_i$, $i=1,2$, on $\WW(\g,f,k)/\partial\WW(\g,f,k)$ induced by the $\lambda$-brackets on $\WW(\g,f,k)$. More precisely,
\begin{equation} \label{Eqn:4.5_0413}
\,\textstyle  [ \int W_1 , \int W_2]'_i = \int \{ W_{1\, \lambda} W_2\}_i|_{\lambda=0}.  
\end{equation}

\begin{prop} \label{Prop:4.13_0804}
 Brackets $[\ , \ ]'_i$ in (\ref{Eqn:4.5_0413}) and  $[\ , \ ]_i$ in (\ref{Eqn:4.2_0413}) for $i=1,2$, defined on $ \WW(\g,f,k)/\partial \WW(g,f,k)$ are same.
\end{prop}

\begin{proof}
By Proposition \ref{Prop:master formula}, for $\phi,\psi\in \WW(\g,f,k)$, we have 
 \begin{equation} \label{Eqn:master_w}
\textstyle \{ \phi_\lambda \psi\}_1 = \sum_{\substack{i,j\in I \\ m,n\in \ZZ_{\geq 0} }}
 C_{i,j}^{\phi,\psi} \frac{\partial_R \psi}{\partial_R q^{ j\, (n)}}(\lambda+\partial)^n \big([q^i , q^j] +(q^i|q^j)k(\lambda+\partial)\big)(-\lambda-\partial)^m \frac{\partial \phi}{\partial q^{i\, (m)}}
 \end{equation}
for  the sign consideration $C_{i,j}^{\phi,\psi}.$ If we apply  $\lambda=0$ to (\ref{Eqn:master_w}) then 
\[ \textstyle \{ \phi_\lambda \psi \}_1 |_{\lambda=0}= \sum_{\substack{i,j\in I \\ m,n\in \ZZ_{\geq 0} }}
 C_{i,j}^{\phi,\psi} \frac{\partial_R \psi}{\partial_R q^{ j\, (n)}}\partial^n \big([q^i , q^j] +(q^i|q^j)k\partial\big)(-\partial)^m \frac{\partial \phi}{\partial q^{i\, (m)}}\]
 and 
 \begin{equation}
 \begin{aligned}
 \textstyle [\int \phi,\int\psi]'_1 & = \textstyle \int  \{ \phi_\lambda \psi\}_1 |_{\lambda=0}\\
  &  \textstyle = \sum_{\substack{i,j\in I \\ m,n\in \ZZ_{\geq 0} }}
 \int C_{i,j}^{\phi,\psi} \big( (-\partial)^n\frac{\partial_R \psi}{\partial_R q^{ j\, (n)}} \big) \big([q^i , q^j] +(q^i|q^j)k\partial\big)(-\partial)^m \frac{\partial \phi}{\partial q^{i\, (m)}}\\
 &\textstyle = \sum_{\substack{i,j\in I \\ m,n\in \ZZ_{\geq 0} }}
 \int \big( (-\partial)^m \frac{\partial_R \phi}{\partial_R q^{i\, (m)}} \big) [q^i , q^j] \big( (-\partial)^n\frac{\partial \psi}{\partial q^{ j\, (n)}} \big)\\
 & + \textstyle  \sum_{\substack{i,j\in I \\ m,n\in \ZZ_{\geq 0} }}
 \int \big( (-\partial)^m \frac{\partial_R \phi}{\partial_R q^{i\, (m)}} \big) (q^i | q^j)k (-\partial)\big( (-\partial)^n\frac{\partial \psi}{\partial q^{ j\, (n)}} \big)\\
 & = \textstyle  \sum_{\substack{i,j\in I \\ m,n\in \ZZ_{\geq 0} }}
 \int  \frac{\delta_R \phi}{\delta_R q^i} ( [q^i, q^j]-(q^i | q^j)k\partial)  \frac{\delta \psi}{\delta q^j}.
 \end{aligned}
 \end{equation}
 Hence $\textstyle [\int \phi,\int\psi]'_1 =\textstyle [\int \phi,\int\psi]_1.$ By same arguments, we have  $\textstyle [\int \phi,\int\psi]'_2 =\textstyle [\int \phi,\int\psi]_2.$
\end{proof}

\begin{thm}
Brackets $[\, , \, ]_1$ and $[\, , \, ]_2$ are Lie brackets on $\int \WW(\g,f,k):= \WW(\g,f,k)/\partial \WW(\g,f,k)$.
\end{thm}

\begin{proof}
We know that if $\{\, _\lambda \, \}$ is a PVA bracket on $\mathcal{P}$ then $\{\, _\lambda \, \}|_{\lambda=0}$ is a Lie algebra bracket on $\mathcal{P}/\partial \mathcal{P}$. Hence the theorem directly follows from Proposition \ref{Prop:4.13_0804}.
\end{proof}

\begin{rem} \label{Rem:4.16_0721}
If we consider $\mathcal{L}(\Lambda)$ instead of $\mathcal{L}^\sigma(\Lambda)$, we have
\begin{equation}
\begin{aligned}
\textstyle\, [\int \phi, \int \psi]_{\mathcal{L}, 1} & :=\textstyle \int \left( \left. \frac{\delta \phi}{\delta q} \right| \left[ \frac{\delta \psi}{\delta q}, \mathcal{L}_{[1]}(\Lambda)\right]\right)  \\
 & =  \textstyle   \sum_{\substack{i,j\in I \\ m,n\in \ZZ_{\geq 0} }}(-1)^{\pa(i)+\pa(j)}
 \int  \frac{\delta_R \phi}{\delta_R q^i}  ( [q^i, q^j]-(q^i | q^j)k\partial)  \frac{\delta \psi}{\delta q^j}
 \end{aligned}
 \end{equation}
 for  bases $\{q_i\}_{i\in I}$ and  $\{q^j\}_{j\in I}$ of $\p$ and $\m^\perp$ such that $(q_i|q^j)= \delta_{ij}.$ 
In this article, we want to discuss integrable systems associated to a W-superalgebra whose PVA structure induces $[\, , \, ]_1$ more than $ [\ , \ ]_{\mathcal{L}, 1}$. Hence we prefer to use $\mathcal{L}^\sigma(\Lambda)$  than $\mathcal{L}(\Lambda)$ (see also the Remark \ref{Rem:5.15_0721}.)
\end{rem}


\section{super-Hamiltonian equations and Poisson vertex algebras}

Let us introduce super-Hamiltonian equations via Poisson vertex algebras. Recall that infinite dimensional Hamiltonian equation on the even differential algebra $\PP=\Cdiff[\, u_i\, |\, i\in I\, ]$ is an evolution equation of the form 
\begin{equation}\label{evenHE}
 \frac{d u}{dt}= H(\partial) \frac{\delta h}{\delta u}
 \end{equation}
where 
\begin{enumerate}
\item
the Poisson operator $H(\partial)= (H_{ij}(\partial))_{i,j\in I}$ is an $|I| \times |I|$ matrix operator such that 
\begin{enumerate}[(i)]
\item $H_{ij}(\partial) = \sum_{n=0}^N H_{ij;n}\partial^n$ for $H_{ij;n}\in \CC[\, \partial^n u_i\, |\, i\in I, n\in \ZZ_{\geq 0}\, ]$,
\item if the $\lambda$-bracket on $\PP$ is defined by $\{u_{i \, \lambda\, } u_j\}_H = \sum_{n=0}^N H_{ij;n} \lambda^n$ then the differential algebra $\PP$ with the induced $\lambda$-bracket $\{\, _\lambda\, \}_H$ is a Poisson vertex algebra,
\end{enumerate}
\item the Hamiltonian $h$ is an element in $\PP$.
\end{enumerate}
Moreover, the equation (\ref{evenHE}) can be written with the $\lambda$-bracket $\{ \, _\lambda\,   \}_H$ as follows:
\[  \frac{d u_i}{dt}= \{ h_\lambda u_i \}_{H} |_{\lambda=0}\, \qquad i\in I. \]
For details, we refer to the paper \cite{BDK}.

Analogously, we define super-Hamiltonian systems and integrable systems when $\PP$ is a differential superalgebra.

\begin{defn} Let $\PP=\Cdiff[\, u_i\, |\, i\in I\, ]$ be the superalgebra of differential  polynomials and let $\PP_{\bar{0}}$ and $\PP_{\bar{1}}$ be the even and odd subspaces of $P$.
\begin{enumerate}
\item A {\it super-Hamiltonian evolution equation} on the differential superalgebra on $\PP$  is an evolution equation of the form 
\begin{equation} \label{Eqn:HE}
 \frac{d \phi}{dt}=  \{\ h\ _\lambda\  \phi \ \}|_{\lambda=0}, \qquad \phi\in \PP 
 \end{equation}
for some $h\in \PP_{\bar{0}}.$
\item An {\it integral of motion} of (\ref{Eqn:HE}) is the local functional $\int f\in \int \PP$ such that 
\[ \small \int \frac{df}{dt}=0.\]
\item If (\ref{Eqn:HE}) has infinitely many linearly independent integrals of motion then it is called an {\it integrable system.}
\end{enumerate}
\end{defn}

From now on, we let  $\PP=\Cdiff[\, u_i\, |\, i\in I\, ]$ be the differential superalgebra  and $u_i$ be homogeneous variables of $\PP$, that is $I= I_{\bar{0}}\cup I_{\bar{1}}$ and $u_i$ for $i\in I_{\bar{0}}$ (resp. $I_{\bar{1}}$) are even (resp. odd.) 

\begin{rem} \ 
 For $f\in \PP$, we let 
\begin{equation}\label{Eqn:HE_2}
\small  \frac{d f}{dt} = \sum_{n\in\ZZ_{\geq 0}, i\in I} \left(\partial^n \frac{d u_i}{dt} \right)\frac{\partial  f}{\partial u_i^{(n)}},
 \end{equation}
inspired from chain rules. Then
 \begin{equation}
\text{ $\left[\  \frac{d u_i}{dt}=  \{\, h\, _\lambda\,  u_i \, \}|_{\lambda=0} \text{ for any } i\in I \  \right]$ iff  $\left[\ \frac{d f}{dt}=  \{\, h\, _\lambda\, f \, \}|_{\lambda=0} \text{ for any } f\in \PP\ \right]$}.
\end{equation}
\end{rem}

\vskip 2mm

For the following proposition, recall if $\PP$ is a PVA  with a $\lambda$-bracket  $\{\, _\lambda\, \}$ then the superspace $\PP/\partial \PP$ is a Lie superalgebra endowed with the bracket $[\int f, \int g] := \int \{ f_\lambda g\}|_{\lambda=0}.$

\begin{prop}[generalized Lenard-Magri scheme] \label{Prop:LM scheme}
Suppose $\PP$ is endowed with two compatible $\lambda$-brackets $\{\, _\lambda\, \}_H$ and $\{\, _\lambda \, \}_K$. If there are linearly independent even elements $\int h_i$, $i\in \ZZ_{\geq 0}$,  in $\int \PP$ such that 
\begin{equation}
 \textstyle [\int h_m, \int u_i]_H=[\int h_{m+1}, \int u_i]_K \quad  \text{ for } m\in \ZZ_{\geq 0} \text{ and } i\in I,
\end{equation} 
then $\frac{d\phi}{dt}= \{ h_{m\, \lambda} \phi\}_{H}|_{\lambda=0}$ for  $m\in \ZZ_{\geq 0}$ are integrable systems and $\int h_{m'}$,  $m'\in \ZZ_{\geq 0}$ are integrals of motion.
\end{prop}

\begin{proof}
 If we assume $m>n$ then 
\[ \textstyle\, [\int h_m, \int h_n]_H= [\int h_m, \int h_{n+1}]_K \text{ and }  \, [\int h_m, \int h_n]_K= [\int h_{m-1}, \int h_n]_H. \]
Inductively, we can prove that if $m-n$ is odd then $[\int h_m, \int h_n]_H= [\int h_l, \int h_l]_K=0$ (resp. $[\int h_m, \int h_n]_K= [\int h_l, \int h_l]_H=0$) for some $l$ with $m\geq l>n$ (resp. $m>l\geq n$). If $m-n$ is even then $[\int h_m, \int h_n]_H= [\int h_l, \int h_l]_H=0$ (resp. $[\int h_m, \int h_n]_K= [\int h_l, \int h_l]_K=0$) for some $l$ with $m> l>n$.
\end{proof}

\begin{rem}
The Lenard-Magri scheme in Proposition \ref{Prop:LM scheme} does not give any clue of finding {\it odd} integrals of motion.
\end{rem}

Consider the adjoint map $\ad \Lambda:  \g((z^{-1})) \to  \g((z^{-1}))$ such that $\ad(\Lambda)( A)= [\Lambda,A]$ for $A\in  \g((z^{-1})).$ In the rest of this section, we assume that $\Lambda$ is semisimple so that 
\begin{equation} \label{Eqn:semi-simple}
 \g((z^{-1})) = \text{ker}(\ad \Lambda) \oplus \text{im}(\ad \Lambda).
 \end{equation}

\begin{rem} \label{Rem:Lambda_assump}
The assumption that $\Lambda$ is a semisimple element is quite a big constraint. Such $\Lambda$ does not always exist for any nilpotent element $f$. However, when $\g$ is $\sll(m|n)$ and the nilpotent element $f$ corresponds to the pair of partitions $\lambda$ and $\mu$ of $m$ and $n$ which has the form of one of followings: 
\begin{equation}
\begin{aligned}
& (1) \ \lambda= (r^{n_{r}}, 1^{n_{1}}), \quad \mu=  (r^{m_{r}}, 1^{m_{1}}), \\
& (2) \ \lambda= ( (r, r-1)^{p_{r,r-1}}, 1^{p_1}), \quad \mu =  ( (r, r-1)^{q_{r,r-1}}, 1^{q_1}),
\end{aligned}
\end{equation}
we can find an even element $s$ such that $\Lambda$ is semisimple. Here, $n_r, n_1, m_r, m_1, p_{r, r-1}, p_1, q_{r,r-1}, q_1$ are all nonnegative integers. Note that this remark directly follows from $\g= \sll_n$ case (see  \cite{DKV} and Remark \ref{Rem:assumption}).
\end{rem}

Take the gradation on $\g((z^{-1}))$ defined by 
\[  \deg(z)= -d-1 \text{ and }  \deg(g)=j/2  \text{ if } g\in\g \text{ and }  [ h/2 ,g]= (j/2) g\]
and denote by $\g((z^{-1}))_k$ the subspace of $\g((z^{-1}))$ consisting of elements with degree $k$.

\begin{prop} \label{prop:fund_Hamil}
Let $L(\Lambda)= k\partial+q-\Lambda\otimes 1$ be a Lax operator. For a subspace $V \subset \g((z^{-1}))$, denote  $V_t= \g((z^{-1}))_t  \cap V$. There exist unique even element  $S(q)\in  \bigoplus_{t>0} \text{im}(\ad \Lambda)_{t}\otimes \mathcal{V}_{\II}(\p)$ and  unique even element $h(q)\in \bigoplus_{t> -1}\text{ker}(\ad \Lambda)_t \otimes \mathcal{V}_{\II}(\p)$ such that 
\begin{equation}\label{Eqn:5.6}
 L_0^{\sigma}(\Lambda):= e^{\ad S(q)^{\sigma}}L^{\sigma}(\Lambda)=k \partial +h(q)^{\sigma}-\Lambda\otimes 1.
 \end{equation}
\end{prop}

\begin{proof}
We can decompose (\ref{Eqn:5.6}) via the gradation on $\g((z^{-1})).$ Then degree $-\frac{1}{2}$ part of (\ref{Eqn:5.6}) is 
\[ q^{\sigma}_{-1/2}+[S^{\sigma}_{1/2},-\Lambda\otimes 1 ]= h^{\sigma}_{-1/2},\]
where $q^{\sigma}_{-1/2}$, $S^{\sigma}_{1/2}$ and $h^{\sigma}_{-1/2}$ are degree $-1/2$, $1/2$ and $-1/2$ part of $q^{\sigma}$, $S(q)^{\sigma}$ and $h(q)^{\sigma}$. By (\ref{Eqn:semi-simple}),  $S^{\sigma}_{1/2}$ and $q^{\sigma}_{-1/2}$ are uniquely determined. Also, since $q^{\sigma}_{-1/2}$ and $\Lambda$ are even, both $S^{\sigma}_{1/2}$ and $q^{\sigma}_{-1/2}$ are even. Inductively, $S^{\sigma}_{n}$ and  $h^{\sigma}_{n-1}$ are uniquely determined for $n\geq 1$. More precisely, if  $S^{\sigma}_{m}$ and $h^{\sigma}_{m-1}$ are known for $m<n$ then degree $n-1$ part of  (\ref{Eqn:5.6}) is 
\[  Q_{n-1}^\sigma+[S^{\sigma}_{n},-\Lambda \otimes 1]= h^{\sigma}_{n-1},\]
where $Q_{n-1}^\sigma$ is then degree $n-1$ part determined by $S^{\sigma}_{m}$'s for $m<n$. Now, by  (\ref{Eqn:semi-simple}),  $S^{\sigma}_{n}$ and $q^{\sigma}_{n-1}$ are uniquely determined. 
\end{proof}
  
\begin{rem} \label{Rem:5.5_1011}
By Proposition \ref{prop:fund_Hamil} and its proof, we can see that for an operator $L(\Lambda)=k \partial+q-\Lambda\otimes 1$, there is  unique  $h(q)\in \bigoplus_{t > -1}\text{ker}(\ad \Lambda)_t \otimes \mathcal{V}_{\II}(\p)$ such that $  L_0^{\sigma}(\Lambda):= e^{\ad S(q)^{\sigma}}L^{\sigma}(\Lambda)=k \partial +h(q)^{\sigma}-\Lambda\otimes 1$ for some $S(q)\in  \bigoplus_{t>0}\g((z^{-1}))_{t}\otimes \mathcal{V}_{\II}(\p)$ (which is not necessarily unique).
\end{rem}
  
For $\mathcal{L}(\Lambda)=k\partial+q_{\text{univ}}-\Lambda\otimes 1$, consider 
 \begin{equation} \label{Eqn:Hamiltonian}
 H_n=(h(q_{\text{univ}})^\sigma|  z^{n} \Lambda\otimes 1)\in \mathcal{V}_{\II}(\p), \quad n \in \ZZ_{\geq 0}.
 \end{equation}
 and fix 
 \begin{equation}\label{Eqn:Hamiltonian-1}
 H_{-1}=((h(q_{\text{univ}})^\sigma-\Lambda\otimes 1 | z^{-1}\Lambda\otimes 1)=(-\Lambda\otimes 1 | z^{-1}\Lambda\otimes 1) \in \CC.
 \end{equation}

 \begin{rem}
 Note that $(h(q_{\text{univ}})^\sigma|  z^{n} \Lambda\otimes 1)\neq 0$ only if $h(q_{\text{univ}})\in \g_{\bar{0}} \otimes \mathcal{V}_\II(\p)_{\bar{0}}$. Hence 
\[  H_n=(h(q_{\text{univ}}) |  z^{n} \Lambda\otimes 1)\in \mathcal{V}_{\II}(\p).\]
\end{rem}

  \begin{lem}
  For $H_n\in \mathcal{V}_{\II}(\p)$, $n\in \ZZ_{\geq -1}$, we have  $H_n\in \WW(\g,f,k).$
  \end{lem}
  
  \begin{proof}
  It is enough to show that $h(q_{\text{univ}})^{\sigma}\in \bigoplus_{t > -1}\text{ker}(\ad \Lambda)_t \otimes \WW(\g,f,k).$ 
  
  Recall that 
 there exists $X^{\sigma}\in(\n \otimes \mathcal{V}_{\II}(\p))_{\bar{0}}$ such that $e^{\ad X^{\sigma}} \mathcal{L}^{\sigma}(\Lambda)=k \partial+\sum_{i\in \mathcal{J}} (q_i\otimes w^i)^\sigma -\Lambda\otimes 1$ where $w^i$ generate $\WW(\g,f,k).$ It is obvious that
 \[ \textstyle h(\sum_{i\in \mathcal{J}} q_i\otimes w^i) \in \bigoplus_{t > -1}\text{ker}(\ad \Lambda)_t \otimes \WW(\g,f,k).\]
More precisely, there is $S \in  \bigoplus_{t > 0}\g((z^{-1}))_t \otimes \mathcal{V}_{\II}(\p)$ such that 
\[ \textstyle e^{\ad (S^{\sigma}+ X^{\sigma})} \mathcal{L}^{\sigma}(\Lambda)=e^{\ad S^{\sigma}}(e^{\ad X^{\sigma}}\mathcal{L}^{\sigma}(\Lambda) )=k \partial+  h\big(\sum_{i\in \mathcal{J}} q_i\otimes w^ i\big)^\sigma -\Lambda\otimes 1.\]
Since $S+X\in \bigoplus_{k > 0}\g((z^{-1}))_k \otimes \mathcal{V}_{\II}(\p)$, by Remark \ref{Rem:5.5_1011}, we conclude that 
 \[\textstyle h(q_{\text{univ}})= h(\sum_{i\in \mathcal{J}} q_i\otimes w^i)\in  \bigoplus_{t > -1}\text{ker}(\ad \Lambda)_t  \otimes \WW(\g,f,k).\]
 Thus  $H_n\in \mathcal{V}_{\II}(\p)=(h(q_{\text{univ}}) |  z^{n} \Lambda\otimes 1)\in \WW(\g,f,k)$.
  \end{proof}
 

 
 \begin{lem} \label{Lem:Hamiltonian}

Let $S$ be the same as in (\ref{Eqn:5.6}) and  $r\in (\m^\perp \otimes \mathcal{V}_\II(\p))_{\bar{0}}$.  Then we have
\begin{equation}\label{Eqn:5.11_0721}
\textstyle   \int   \left( r^{\sigma} \left| \frac{\delta H_n}{\delta q}\right.\right)=\int \left( r^{\sigma} \left|  e^{-\ad S^{\sigma}}( z^{n} \Lambda \otimes 1)\right.\right), \quad n\in \ZZ_{\geq -1},
  \end{equation}
where $q=(q^i)_{i\in I}$ for a basis $q^i$ of $\p$. 
 \end{lem}
 
 \begin{proof}
 If $n= -1$, we have 
 \[ e^{-\ad S^{\sigma}}( z^{n} \Lambda \otimes 1)\in \g [z] z \otimes \mathcal{V}_\II(\p)\]
 so that $\left( r^{\sigma} \left|  e^{-\ad S^{\sigma}}( z^{n} \Lambda \otimes 1)\right.\right)=0.$ Since $\frac{\delta H_{-1}}{\delta q}=0$, the both sides of (\ref{Eqn:5.11_0721}) are 0.
 
To show the proposition when $n\geq 0$, recall that 
 \[  \textstyle \int  \left.\frac{H_n(q+\epsilon r)-H_n(q)}{\epsilon}\right|_{\epsilon=0}= \int \left( r^{\sigma} \left| \frac{\delta H_n}{\delta q}\right.\right) \]
 for $r\in (\m^\perp \otimes \mathcal{V}(\p))_{\bar{0}}$ and $q= \sum_{i\in I} q_i \otimes Q^i$. Let us define $L_0(\epsilon)$ and $S(\epsilon)$ by 
 \[ L_0^{\sigma}(\epsilon) =e^{\ad S(\epsilon)^{\sigma}} (\partial+(q^{\sigma}+\epsilon r^{\sigma})-\Lambda \otimes 1)=  \partial +h(q+\epsilon r)^{\sigma} -\Lambda \otimes 1.\]
 Then 
 \[ \textstyle \left.\left( \left. \frac{d}{d\epsilon} L_0^{\sigma}(\epsilon)\right|  z^{n}\Lambda\otimes 1 \right) \right|_{\epsilon=0} = \left.\frac{H_n(q+\epsilon r)-H_n(q)}{\epsilon}\right|_{\epsilon=0}.\]
 By the process finding $S(\epsilon)$ in Proposition \ref{prop:fund_Hamil}, we have 
 \[ S(\epsilon)= S_0+ \epsilon S_1+\epsilon^2 S_2 +\cdots \in    \text{im}(\ad \Lambda)((z^{-1})) \otimes \mathcal{V}_\II(\p) [\epsilon],\]
 for $S_0, S_1, \cdots \in\g((z^{-1})) \otimes \mathcal{V}_\II(\p).$
Also, we have 
\begin{equation} \label{Eqn:keyequation}
  \textstyle \left. \frac{d}{d\epsilon} L_0^{\sigma}(\epsilon) \right|_{\epsilon=0} = \left. e^{\ad S(\epsilon)^{\sigma}} r^{\sigma} \right|_{\epsilon=0}+ \left. \frac{d}{d\epsilon} \left(  e^{\ad S(\epsilon)^\sigma} L^{\sigma}(0) \right) \right|_{\epsilon=0}
\end{equation}
 where $L^{\sigma}(0)= \partial + q^{\sigma}-\Lambda\otimes 1$ and $e^{\ad S_0^{\sigma}} L^{\sigma}(0)=L_0^{\sigma}(0)=\partial +h^{\sigma}(q)-\Lambda\otimes 1.$ In order to investigate the last term  in (\ref{Eqn:keyequation}), we need the following facts:
 \begin{equation} \label{comp}
 \begin{aligned}
& (i)\ \textstyle  \left. \frac{d}{d\epsilon} \left(  e^{\ad S(\epsilon)^\sigma} L^{\sigma}(0) \right) \right|_{\epsilon=0}
 = \sum_{n\in \ZZ_{\geq 1}}\frac{1}{n!} \sum_{m=0}^{n-1} \ad^m S^\sigma_0(\ad S^\sigma_1(\ad ^{n-1-m}S^\sigma_0 (L^\sigma(0)))), \\
& (ii) \  \textstyle \frac{1}{n!} \sum_{m=0}^{n-1} \ad^m S^\sigma_0(\ad S^\sigma_1(\ad ^{n-1-m}S^\sigma_0 (L^\sigma(0)))) \\
& \textstyle \qquad \qquad = \sum_{m=0}^{n-1} \left[ \, \frac{1}{(n-m)!}\ad^{n-m-1}S^\sigma_0 (S^\sigma_1)\, , \, \frac{1}{m!} \ad^m S^\sigma_0(L^\sigma(0)) \, \right], \\
& (iii) \ \textstyle \sum_{n\in \ZZ_{\geq 1}}\sum_{m=0}^{n-1} \left[ \, \frac{1}{(n-m)!}\ad^{n-m-1}S^\sigma_0 (S^\sigma_1)\, , \, \frac{1}{m!} \ad^m S^\sigma_0(L^\sigma(0)) \, \right]\\
& \textstyle \qquad \qquad  = \sum_{l\in \ZZ_{\geq 0}}\left[ \frac{1}{(l+1)!} \ad^l S_0^\sigma (S_1^\sigma), L_0^\sigma(0)\right].
\end{aligned}
 \end{equation} 
 By (\ref{Eqn:keyequation}) and (\ref{comp}), we have 
 \[\textstyle  \textstyle  \left. \frac{d}{d\epsilon} \left(  e^{\ad S(\epsilon)^\sigma} L^{\sigma}(0) \right) \right|_{\epsilon=0}= \sum_{l\in \ZZ_{\geq 0}}\left[ \frac{1}{(l+1)!} \ad^l S_0^\sigma (S_1^\sigma), L_0^\sigma(0)\right]. \] 

Observe that 
\begin{equation}
 \textstyle  \int \left( r^{\sigma} \left| \frac{\delta H_n}{\delta q}\right.\right) = \int \left.\left( \left. \frac{d}{d\epsilon} L_0^{\sigma}(\epsilon)\right|  z^{n}\Lambda\otimes 1 \right) \right|_{\epsilon=0}
\end{equation}
and 
 \begin{equation} \label{Eqn:5.8}
 \begin{aligned}
&\left.\left( \left. \frac{d}{d\epsilon} L_0^{\sigma}(\epsilon)\right|  z^{n}\Lambda\otimes 1 \right) \right|_{\epsilon=0}\\
& \textstyle  =  \left.\left(\left. e^{\ad S(\epsilon)^{\sigma}} r^{\sigma}\right| z^{n} \Lambda\otimes 1 \right)\right|_{\epsilon=0} + \sum_{l\in \ZZ_{\geq 0}}\frac{1}{(l+1)!} \left. \left(\left. \left[ \ad^l S_0^\sigma (S_1^\sigma), L_0^\sigma(0)\right]\right| z^{n} \Lambda\otimes 1 \right)\right|_{\epsilon=0}\\
& \textstyle = \left( r^{\sigma} \left|  e^{-\ad S(0)^{\sigma}}( z^{n} \Lambda\otimes 1)\right.\right)+\sum_{l\in \ZZ_{\geq 0}}\frac{1}{(l+1)!}\left. \left(\left. \ad^l S_0^\sigma (S_1^\sigma)\right| [h^\sigma(q)-\Lambda\otimes 1 ,z^{n}\Lambda\otimes 1]\right)\right|_{\epsilon=0}\\
& \textstyle \qquad \qquad \qquad +\sum_{l\in \ZZ_{\geq 0}}\frac{1}{(l+1)!}\left. \left(\left. -\partial (\ad^l S_0^\sigma (S_1^\sigma) ) \right| z^{n}\Lambda\otimes 1 \right)\right|_{\epsilon=0}\\
& = \textstyle  \left( r^{\sigma} \left|  e^{-\ad S(0)^{\sigma}}( z^{n} \Lambda\otimes 1)\right.\right)+
\sum_{l\in \ZZ_{\geq 0}}\frac{1}{(l+1)!}\left. \left(\left. -\partial (\ad^l S_0^\sigma (S_1^\sigma) ) \right| z^{n}\Lambda\otimes 1 \right)\right|_{\epsilon=0}.
 \end{aligned}
 \end{equation}
The last term in (\ref{Eqn:5.8}) is in $\partial \mathcal{V}_{\II}(\p)$.  
Hence, for $S(0)^\sigma= S^\sigma$,  we have
\[ \textstyle \int   \left( r^{\sigma} \left| \frac{\delta H_n}{\delta q}\right.\right)= \int \left( r^{\sigma} \left|  e^{-\ad S^{\sigma}}( z^{n} \Lambda\otimes 1)\right.\right).\]
 \end{proof}
  
 \begin{prop} \label{Prop:5.7_0609}
  Recall that $\g=\m \oplus \p$ and suppose $\{q^i|i\in I\}$ is a basis of $\p$ and $\{q^i|i\in I'\}$ is a basis of $\m$.  For $ Q= \sum_{i\in I}  q^i \otimes Q_i+ \sum_{i\in I'} q^i\otimes Q_i \in \g\otimes  \mathcal{V}(\p) $, we denote $Q |_\p= \sum_{i\in I} q^i\otimes Q_i\in\p \otimes  \mathcal{V}(\p)$. Then we have
  \[   \frac{\delta H_n}{\delta q}=\left. e^{-\ad S^{\sigma}}( z^{n} \Lambda\otimes 1)\right|_\p.\]
 \end{prop}
  
  \begin{proof}
  It directly follows from Lemma \ref{Lem:Hamiltonian}.
  \end{proof}
  
  \begin{prop} \label{Prop:5.8_0609}
  Let $\phi\in \WW(\g,f,k)$ and $a_\m\otimes F \in (\m\otimes \mathcal{V}_\II(\p)_{\bar{0}}.$ Then we have
  \[  \textstyle  \int \left( \left. \frac{\delta \phi}{\delta q} \right|  [(a_\m\otimes F)^\sigma, \mathcal{L}^{\sigma}_{[1]}(\Lambda)] \right)=  \int \left( \left. \frac{\delta \phi}{\delta q} \right| z^{-1} [ (a_\m \otimes F)^\sigma, zs\otimes 1] \right)= 0.\] 
  \end{prop}
  
  \begin{proof}

  Let us denote $S= a_\m\otimes F\in (\m\otimes \mathcal{V}_\II(\p))_{\bar{0}}$ and denote 
  \[ \mathcal{L}^{\sigma}(\epsilon) := e^{\ad  \epsilon S^{\sigma}} \mathcal{L}^{\sigma}(\Lambda).\]
  Then $\mathcal{L}^{\sigma}(\epsilon)= \partial + q^{\sigma}(\epsilon) - \Lambda\otimes 1$ for $q^{\sigma}(\epsilon)=q_{\text{univ}}^{\sigma} +\epsilon [S^{\sigma}, \mathcal{L}^{\sigma}(\Lambda)] + o (\epsilon^2).$ Since $\phi$ is gauge invariant, we have
  \[  \textstyle 0= \left.  \int \frac{d \phi(q(\epsilon))}{d\epsilon}\right|_{\epsilon=0}= \int \left( \left. \frac{\delta \phi}{\delta q} \right| [ (a_\m\otimes F)^\sigma, \mathcal{L}^\sigma(\Lambda)] \right) \]
  for any $F$ which has the same parity as $a_\m.$ Hence $\int \left( \left. \frac{\delta \phi}{\delta q} \right|  [(a_\m\otimes F)^\sigma, \mathcal{L}^{\sigma}_{[1]}(\Lambda)] \right)=0.$
  
  Also, the second equality follows from $[a_\m, s]= 0.$
  \end{proof}
  
  \begin{prop} \label{Prop:5.9_0613}
  Let $S$ be the same as in (\ref{Eqn:5.6}) and let $\phi\in \WW(\g,f,k)$. Then we have
  \begin{equation}
  \begin{aligned}
  & \textstyle \int \left( \left. \frac{\delta \phi}{\delta q} \right| \left[ \frac{\delta H_n}{\delta q}, \mathcal{L}^{\sigma}_{[1]}(\Lambda)\right] \right)= \int  \left( \left. \frac{\delta \phi}{\delta q} \right| \left[ e^{-\ad S^{\sigma}}( z^{n} \Lambda\otimes 1), \mathcal{L}^{\sigma}_{[1]}(\Lambda)\right] \right),\\
   & \textstyle  \int \left( \left. \frac{\delta \phi}{\delta q} \right| \left[ \frac{\delta H_n}{\delta q},  s\otimes 1 \right] \right)= \int  \left( \left. \frac{\delta \phi}{\delta q} \right| \left[ e^{-\ad S^{\sigma}}( z^{n} \Lambda\otimes 1),  s\otimes 1\right] \right).
  \end{aligned}
  \end{equation}
  \end{prop}
  
  \begin{proof}
  It follows from Proposition \ref{Prop:5.7_0609} and Proposition \ref{Prop:5.8_0609}.
  \end{proof}
  
 \begin{thm} \label{Thm:5.10_0730}
 Let us consider $\WW(\g,f,k)$ and $H_i$ be defined as in (\ref{Eqn:Hamiltonian}). The equation 
 \begin{equation}\label{Eqn:Walg_HE}
  \frac{du}{dt}= \{ H_{i\, \lambda\, } u\}_1|_{\lambda=0}, \qquad u\in \WW(\g,f,k) \text{ and } i\in \ZZ
  \end{equation}
 has linearly independent  integrals of motion $\int H_j$ for $j\in \ZZ_{\geq 0}$. Hence (\ref{Eqn:Walg_HE}) is an integrable system.
 \end{thm}
  
  \begin{proof}
  In order to use Lenard scheme, we aim to show that 
  \begin{equation} \label{Eqn:Lenard}
   \textstyle[ \int H_i, \int u]_1= [ \int H_{i+1}, \int u]_2, \quad i\in \ZZ_{\geq -1}.
   \end{equation}
 Recall that the brackets $[ \int H_i, \int u]_1$ and $[ \int H_{i-1}, \int u]_2$ are defined by 
  \begin{equation}\label{Eqn:5.13}
  \begin{aligned}
  &  \textstyle [\int H_i, \int u]_1 = \int \left( \frac{\delta H_i}{\delta q}  \left| \left[\frac{\delta u}{\delta q}, \mathcal{L}^{\sigma}_{[1]}(\Lambda)\right] \right. \right) 
  = -  \int \left( \frac{\delta u}{\delta q}  \left| \left[ \frac{\delta H_i}{\delta q}, \mathcal{L}^{\sigma}_{[1]}(\Lambda)\right] \right. \right), \\
&  \textstyle [\int H_{i+1}, \int u]_2 = \int \left( \frac{\delta H_{i+1}}{\delta q}  \left| \left[\frac{\delta u}{\delta q},  s\otimes 1 \right] \right. \right) 
  =  -\int \left( \frac{\delta u}{\delta q}  \left| \left[ \frac{\delta H_{i+1}}{\delta q}, s\otimes 1 \right] \right. \right).
\end{aligned}
\end{equation}
Denote $\frac{\delta H(z)}{\delta q}:=\sum_{i\in \ZZ} \frac{\delta H_i}{\delta q}z^{-i}$. Then 
\begin{equation}\label{Eqn:5.14}
\textstyle  \left[ \frac{\delta H(z)}{\delta q}, \mathcal{L}^{\sigma}(\Lambda)\right]= \sum_{i\in \ZZ} \left( \left[ \frac{\delta H_i}{\delta q} , \mathcal{L}^{\sigma}_{[1]}(\Lambda)\right] - \left[ \frac{\delta H_{i+1}}{\delta q} ,  s\otimes 1 \right] \right)z^{-i}. 
\end{equation}
By (\ref{Eqn:5.13}) and (\ref{Eqn:5.14}), (\ref{Eqn:Lenard}) is equivalent to 
\begin{equation}
\begin{aligned}
&  \textstyle   \int  \left( \left. \frac{\delta u }{\delta q} \right| z^{i}\left[ \frac{\delta H(z)}{\delta q}, \mathcal{L}^{\sigma}(\Lambda)\right] \right) \\
&  =\textstyle \int \left( \left. \frac{\delta u}{\delta q} \right| \left[ e^{-\ad S^{\sigma}}( z^{i} \Lambda\otimes 1), \mathcal{L}_{[1]}^{\sigma}(\Lambda)\right] \right)+ \int \left( \left. \frac{\delta u}{\delta q} \right| \left[ z^{-1}e^{-\ad S^{\sigma}}( z^{i+1} \Lambda\otimes 1), z s\otimes 1\right] \right)
 \\ 
 & \textstyle =\int \left( \left. \frac{\delta u}{\delta q} \right| \left[ e^{-\ad S^{\sigma}}( z^{i} \Lambda\otimes 1), \mathcal{L}^{\sigma}(\Lambda)\right] \right)=0
\end{aligned}
\end{equation}
for any $\phi\in \WW(\g,f,k)$ and  $i\in \ZZ.$ Here, we used  Proposition \ref{Prop:5.9_0613} for the first equality.
Hence we proved (\ref{Eqn:5.13}) by the following fact: 
\[ \left[ e^{-\ad S^{\sigma}}( z^{i} \Lambda\otimes 1), \mathcal{L}^{\sigma}(\Lambda)\right]= e^{-\ad S^{\sigma}}\left[  z^{i} \Lambda\otimes 1, e^{\ad S^{\sigma}}\mathcal{L}^{\sigma}(\Lambda)\right]=0.\]
In particular, since $H_{-1}$ is constant, we have $[\int H_{-1}, \int u]_1=[\int H_0, \int u]_2=0$.

Now, the only thing to prove is that $\{\int H_j\}_{j\in \ZZ_{\geq 0}}$ is linearly independent. Since, for given $H, u\in \WW(g,f,k)$ such that $\{ H_\lambda u\}_1 |_{\lambda=0}\neq 0$, the total degree of $\{ H_\lambda u\}_1 |_{\lambda=0}$ is greater than the total degree of $\{ H_\lambda u\}_2 |_{\lambda=0}$ in the algebra of polynomials  $\CC[ (q^i)^{(n)} | i\in I,\, n\in \ZZ_{\geq 0}],$ where $\{q^i|i\in I\}$ is a basis of $\p$. Hence, if we can show  $\{H_{0\, \lambda} \WW(g,f,k)\}|_{\lambda=0}\neq 0$ then the linearly independence of $\{\int H_j\}_{j\in \ZZ_{\geq 0}}$ follows. Indeed, this can be proved as below. Suppose
$q^{\text{can}}_{\text{univ}}$ in Proposition \ref{Prop:3.14_0730} has the summand $f \otimes \phi_e$. Then $H_0= \phi_e$ so that $\{H_{0\, \lambda} \WW(g,f,k)\}|_{\lambda=0}\neq 0.$ 

  \end{proof}

\begin{rem} \label{Rem:5.15_0721}
Recall that the formula (\ref{Eqn:5.6}) $e^{\ad S(q)^\sigma}L^\sigma(\Lambda)$ is same as $(e^{S(q)}L(\Lambda))^\sigma$. Also, since $\Lambda$ is even, we have 
\[(h(q_{\text{univ}})^\sigma| z^{-n} \Lambda\otimes 1)= (h(q_{\text{univ}})| z^{-n} \Lambda\otimes 1).\]
Hence we can use $L(\Lambda)$ instead of $L(\Lambda)^\sigma$ to compute $H_n$.
\end{rem}
  
  \begin{ex}
  As in Example \ref{Ex:spo(1|2)W}, the Lie superalgebra $\g=\text{spo}(2|1)$ is generated by $e$,  $e_{od}$, $h$, $f_{od}$ and $f$. For $\Lambda= f+ze$ and $K:= -f+ze$, we can see that $\g((z^{-1}))$ is the $\CC((z^{-1}))$-module generated by $e_{od}, f_{od}, h=2x, \Lambda, K$. The subspace $\text{im}(\ad \Lambda)$ is generated by $e_{od}, f_{od}, h, K$ and $ker(\ad \Lambda)$ is generated by $\Lambda$. Note that the Lie brackets between generators are 
  \[ \, [x, \Lambda]= K, \quad [f_{od}, \Lambda]=-z e_{od}, \quad [e_{od}, \Lambda]= -f_{od}, \quad [K, \Lambda]= 2zh=4zx, \quad [x,K]= \Lambda.\]
  Consider the operator,
  \[ L(\Lambda)=k \partial+e_{od}\otimes \phi_{\text{od}}+ e\otimes \phi_{\text{ev}}- \Lambda \otimes 1\]
  for the generators $\phi_{\text{od}}=-\frac{1}{2} f_{od} +\frac{k}{2} \partial e_{od} +\frac{1}{4} h e_{od}$ and $\phi_{\text{ev}}= f+\frac{1}{2} f_{od}e_{od}-\frac{1}{4} h^2 +\frac{k}{4} e_{od} \partial e_{od} -\frac{k}{2} \partial h$ of the algebra in Example \ref{Ex:spo(1|2)W}.
 We want to find $h\in \bigoplus_{t > -1}\text{ker}(\ad \Lambda)_t \otimes \mathcal{V}_{\II}(\p)$ such that 
\begin{equation}\label{Equation:5.16}
L_0(\Lambda):= e^{\ad S}L(\Lambda)= \partial +h-\Lambda\otimes 1.
 \end{equation}
 for some $S\in  \bigoplus_{t>0} \text{im}(\ad \Lambda)_{t}\otimes \mathcal{V}_{\II}(\p)$. 
 Let us denote $U= \sum_{t>0} U_t$ for $U_t\in \mathcal{V}_{\II}(\p)\otimes \g((z^{-1}))_t$ and $h= \sum_{t>-1} h_t$ for $h_t=  \mathcal{V}_{\II}(\p)\otimes \left(  \g((z^{-1}))_t \cap ker(\ad \Lambda) \right)$. Since $z^i \Lambda$ has degree $-1-2i$, we have $h= \sum_{t\in \ZZ_{>-1}}h_{2t+1}.$

 By comparing degree $\frac{1}{2}$ part of (\ref{Equation:5.16}), we have  $S_{1/2}=S_{1}=0$ and 
\[ e_{od}\otimes \phi_{od}- [S_{3/2}, \Lambda\otimes 1]=0.\]
Hence $S_{3/2}=-z^{-1} f_{od} \otimes \phi_{od}$.   By comparing degree $1$ part of (\ref{Equation:5.16}), we have 
\[ e\otimes \phi_{ev} - [S_2, \Lambda\otimes 1]= h_1\] 
Hence $S_2= \frac{1}{2} z^{-1} x \otimes \phi_{ev}$ and $h_1= \frac{1}{2} z^{-1}\Lambda \otimes \phi_{ev}$ so that 
\[ H_0= (h_1| \Lambda\otimes 1)= \phi_{ev}.\]

By degree $\frac{3}{2}$ and $2$ parts of (\ref{Equation:5.16}), we have 
\begin{equation*}
\begin{aligned}
& -[S_{5/2}, \Lambda\otimes 1] + [S_{3/2}, k\partial] =0 \quad \text{ and } \quad S_{5/2}= -z^{-1}e_{od} \otimes k\partial \phi_{od};\\
& -[S_3,  \Lambda\otimes 1]+[S_2, k\partial]+ [S_{3/2},  e_{od}\otimes \phi_{od}]- [S_{3/2}, [S_{3/2}, \Lambda\otimes 1]]=0
\end{aligned}
\end{equation*}
and $S_3= -\frac{1}{8}z^{-2} K \otimes k\partial \phi_{ev}.$ By degree $3$ part of  (\ref{Equation:5.16}), we have
\begin{equation} \label{Eqn:5.17}
\begin{aligned}
& -[S_4, \Lambda\otimes 1] + [S_3, k\partial] + [ S_{5/2}, e_{od}\otimes \phi_{od}]+[S_2, e\otimes \phi_{ev}] \\
& -\frac{1}{2} [S_2, [S_2,  \Lambda\otimes 1] ]-\frac{1}{2} [S_{3/2}, [S_{5/2},  \Lambda\otimes 1]] -\frac{1}{2} [S_{5/2}, [S_{3/2},  \Lambda\otimes 1]]\\
& +\frac{1}{2} [S_{3/2}, [S_{3/2}, k\partial]]= h_3
\end{aligned}
\end{equation}
Since the LHS of (\ref{Eqn:5.17}) is 
\[  -[S_4, \Lambda\otimes 1] +  z^{-2}\Lambda\otimes  \left( \frac{1}{8}\phi_{ev}^2+\frac{k}{2} \partial \phi_{od}\phi_{od} \right) +z^{-2}K \otimes \left( \frac{k^2}{8}\partial^2 \phi_{ev}+\frac{1}{4} \phi_{ev}^2+\frac{k}{2} \partial \phi_{od}\phi_{od}\right) \]
and $K$ is in the image of $\ad \Lambda$, we can conclude  $h_3= z^{-2} \Lambda \otimes \left( \frac{1}{8}\phi_{ev}^2+\frac{k}{2} \partial \phi_{od}\phi_{od} \right)$ so that
\[  H_1= (h_3|z \Lambda \otimes 1)= \frac{1}{4}\phi_{ev}^2+ k \partial \phi_{od}\phi_{od}. \]
Note that $h_3$ is the first integral of motion which gives rise to a nonlinear Hamiltonian equation.
Now we get
\begin{equation*}
\left\{
\begin{array}{l}
\begin{aligned}
 \frac{d\phi_{ev}}{dt}  &  = \big\{  H_{1 \, \lambda} \phi_{ev}\big\}_{H}|_{\lambda=0} = -\frac{k^3}{4}\partial^3 \phi_{ev} -\frac{3k}{2}\partial \phi_{ev}\phi_{ev}+3k^2\phi_{od}\partial^2\phi_{od}, \\
\frac{d\phi_{od}}{dt}  & = \big\{  H_{1 \, \lambda} \phi_{od}\big\}_{H}|_{\lambda=0} = k^3 \partial^3 \phi_{od} -\frac{3k}{2} \partial \phi_{od}\phi_{ev} - \frac{3k}{4}\phi_{od}\partial \phi_{ev}, 
\end{aligned}
\end{array}\right.
\end{equation*}
which is same as super-KdV equation in \cite{KZ} up to constant factors.
\end{ex}

\begin{ex}
Let $\g= \sll(2|1)$ and $f=e_{21}$. With the notations in Example \ref{Ex:W(sl(2|1))}, we consider the operator
\[  L(\Lambda)= k\partial+ \tau\otimes \phi_\tau + e_{1\bar{1}}\otimes \phi_{1\bar{1}}+ e_{\bar{1}2}\otimes \phi_{\bar{1}2}+e\otimes \phi_e -\Lambda \otimes 1.\]
Recall that $\phi_e= -k L-2\phi_t^2$ for the energy momentum field $L$. It is not hard to check that  
\[  \{\phi_{\tau\, \lambda} \phi_{\tau}\}= -\frac{1}{2} k\lambda, \ \{ \phi_{\tau\, \lambda}\phi_{1\bar{1}}\}= -\frac{1}{2} \phi_{1\bar{1}},\  \{\phi_{\tau\, \lambda} \phi_{\bar{1}2}\}= \frac{1}{2} \phi_{\bar{1}2}, \ \{\phi_{1\bar{1}\, \lambda}\phi_{1\bar{1}}\}= \{ \phi_{\bar{1}2\, \lambda} \phi_{\bar{1}2}\}=0.\]
Also, 
\[ \{\phi_{1\bar{1}\, \lambda} \phi_{\bar{1}2}\}= k L-k\partial \phi_\tau - 2k\lambda\phi_\tau-k^2\lambda^2, \ \ \{\phi_{\bar{1}2\, \lambda} \phi_{1\bar{1}}\}= kL+k\partial \phi_\tau+2k\lambda \phi_\tau-k^2\lambda ^2,\]
for the energy momentum field $L$ in Proposition \ref{Prop:2.11_0731}. Recall that $\phi_e= -kL-2\phi_\tau^2$ and  $\lambda$-brackets between $\phi_e$ and elements in $\WW(\g,f,k)$ can be computed using 
\[  \{L_\lambda \phi_\tau\}= (\partial+\lambda)\phi_\tau, \quad \{L_\lambda \phi_{1\bar{1}}\}= (\partial+\frac{3}{2}\lambda)\phi_{1\bar{1}}\quad  \{L_\lambda \phi_{\bar{1}2}\}= (\partial+\frac{3}{2}\lambda)\phi_{\bar{1}2}, \]
and 
\[   \{L_\lambda L\}= (\partial+2\lambda)L-\frac{1}{2} k \lambda^3.\]
Now, let us consider $S\in \bigoplus_{t>0} \text{im} (\ad \Lambda)_t \otimes \mathcal{V}_{\II} (\p)$ and $h\in \bigoplus_{t>-1} \ker(\ad \Lambda)_t \otimes \mathcal{V}_{\II}(\p)$ such that 
\begin{equation} \label{5.23_0722}
  L_0(\Lambda)= k\lambda + h -\Lambda\otimes 1 = e^{\ad S} L(\Lambda).
  \end{equation}
By equating degree $\leq 1/2$-parts of the both sides of (\ref{5.23_0722}), we get 
\[  S_{1/2}= S_1=0, \text{ and } S_{3/2}= -z^{-1}e_{2\bar{1}}\otimes \phi_{1\bar{1}}+z^{-1}e_{\bar{1}1}\otimes \phi_{\bar{1}2}.\]
By equating degree $1$-parts of the both sides of (\ref{5.23_0722}):
\[ \,   [S_2, \Lambda\otimes 1]+h_1 = e\otimes \phi_e = \frac{1}{2} z^{-1} \Lambda \otimes \phi_e+\frac{1}{2}z^{-1} K \otimes \phi_e,\]
for $K=ze-f$, we get
\[   S_2= \frac{1}{2} z^{-1} x \otimes \phi_e, \quad h_1=  \frac{1}{2} z^{-1} \Lambda\otimes \phi_e, \quad H_0= (h_1|\Lambda\otimes 1)= \phi_e .\]
Hence 
\begin{equation}
\left\{ \begin{array}{l}
\begin{aligned}
\frac{d\phi_{\tau}}{dt} &  =\{ \phi_{e\, \lambda} \phi_\tau\}|_{\lambda=0} =\left. k(\lambda+\partial) \phi_\tau \right|_{\lambda=0}=k\partial\phi_\tau,\\
\frac{d\phi_{1\bar{1}}}{dt} & =\{ \phi_{e\, \lambda} \phi_{1\bar{1}}\}|_{\lambda=0}= \left.-k(\partial+\frac{3}{2}\lambda)\phi_{1\bar{1}}+2\phi_{1\bar{1}}\phi_\tau\right|_{\lambda=0}=-k\partial\phi_{1\bar{1}}+2\phi_{1\bar{1}}\phi_\tau,\\
\frac{d\phi_{\bar{1}2}}{dt} & =\{ \phi_{e\, \lambda} \phi_{\bar{1}2}\}|_{\lambda=0} = \left. -k(\partial+\frac{3}{2}\lambda)\phi_{\bar{1}2}-2\phi_{\bar{1}2}\phi_\tau \right|_{\lambda=0}=-k\partial\phi_{\bar{1}2}-2\phi_{\bar{1}2}\phi_\tau,\\
\frac{d\phi_e}{dt} & = \{ \phi_{e\, \lambda} \phi_e\}|_{\lambda=0} = \left. k^2(\partial+2\lambda)L-\frac{k^3}{2} \lambda^3 \right|_{\lambda=0}= k^2\partial L=-k\partial (\phi_e+2\phi_\tau^2),
\end{aligned}
\end{array}\right.
\end{equation}
is the simplest integrable system associated to $\sll(2|1).$

The degree $2$-part of (\ref{5.23_0722}) is 
\begin{equation}
\begin{aligned}
\ [S_{5/2}, \Lambda \otimes 1] & = [ S_{3/2}, k\partial +\tau\otimes \phi_\tau]\\
& = z^{-1}e_{2\bar{1}} \otimes (k\partial \phi_{1\bar{1}}-\phi_\tau \phi_{1\bar{1}})+ z^{-1} e_{\bar{1}1} \otimes (-k\partial \phi_{\bar{1}2}-\phi_\tau\phi_{\bar{1}2}).
\end{aligned}
\end{equation}
Hence $S_{5/2}= z^{-1}e_{1\bar{1}} \otimes (-k\partial \phi_{1\bar{1}}+\phi_\tau \phi_{1\bar{1}})+ z^{-1} e_{\bar{1}2} \otimes (-k\partial \phi_{\bar{1}2}-\phi_\tau\phi_{\bar{1}2}).$

The degree $2$-part of (\ref{5.23_0722}) is 
\begin{equation}
\begin{aligned}
\  &[S_3, \Lambda\otimes 1]+h_2 \\
& = [ S_2, k\partial +\tau\otimes \phi_\tau]+ [ S_{3/2}, e_{1\bar{1}}\otimes \phi_{1\bar{1}} + e_{\bar{1}2} \otimes \phi_{\bar{1}2}]- \frac{1}{2} [S_{3/2}, [S_{3/2}, \Lambda \otimes 1]].
\end{aligned}
\end{equation}
Hence $S_3= z^{-2} K \otimes -\frac{k}{8} \partial \phi_e$ and $h_2= z^{-1} \tau \otimes \frac{1}{2} \phi_{1\bar{1}}\phi_{\bar{1}2}$.

The degree $3$-part of (\ref{5.23_0722}) is 
\begin{equation}
\begin{aligned}
& [S_4, \Lambda \otimes 1] +h_3\\
&= [S_3, k\partial +\tau \otimes \phi_\tau]+[S_{5/2}, e_{1\bar{1}}\otimes \phi_{1\bar{1}}+e_{\bar{1}2}\otimes \phi_{\bar{1}2}] + \frac{1}{2} [S_{3/2}, [S_{3/2}, k\partial +\tau\otimes \phi_\tau]]\\
&-\frac{1}{2} [S_2, [S_2,\Lambda\otimes 1]]-\frac{1}{2} [S_{3/2}, [S_{5/2},\Lambda\otimes 1]]-\frac{1}{2} [S_{5/2}, [S_{3/2},\Lambda\otimes 1]].
\end{aligned}
\end{equation}
Hence $S_4= z^2 h \otimes \frac{1}{2}\left(\frac{k^2}{8} \partial^2 \phi_e +\frac{1}{2} \phi_\tau \phi_{\bar{1}2}\phi_{1\bar{1}} +\frac{k}{2} \partial \phi_{\bar{1}2}\phi_{1\bar{1}}-\frac{k}{2} \phi_{\bar{1}2}\partial \phi_{1\bar{1}}\right)$ and 
\[ h_3= z^{-1} \Lambda\otimes \left( -\frac{1}{8} \phi_e^2 +\frac{1}{2} \phi_\tau \phi_{\bar{1}2}\phi_{1\bar{1}} +\frac{k}{2} \partial \phi_{\bar{1}2}\phi_{1\bar{1}}-\frac{k}{2} \phi_{\bar{1}2}\partial \phi_{1\bar{1}}\right)\]
so that 
\[ H_1=(h_3|z \Lambda \otimes 1)= -\frac{1}{4} \phi_e^2 + \phi_\tau \phi_{\bar{1}2}\phi_{1\bar{1}} +k \partial \phi_{\bar{1}2}\phi_{1\bar{1}}- k \phi_{\bar{1}2}\partial \phi_{1\bar{1}}.\]

Hence the second integrable system associated to $\sll_2$, getting by the Hamiltonian $4H_1$, is 
\begin{equation}
\left\{ \begin{array}{l}
\begin{aligned}
 \frac{d\phi_\tau}{dt}  = & 6k (\partial \phi_{1\bar{1}}\phi_{\bar{1}2}+\phi_{1\bar{1}}\partial \phi_{\bar{1}2}),\\
 \frac{d\phi_{1\bar{1}}}{dt} = &  \phi_{1\bar{1}}( 8\phi_\tau^3-8\phi_e\phi_\tau-12k \phi_\tau\partial \phi_\tau -3k\partial \phi_3+4k^2 \partial^2 \phi_\tau),\\
 &+ \partial \phi_{1\bar{1}}(6k\phi_e-24k\phi_\tau^2 +16 k^2 \partial \phi_\tau)+20 k^2 \partial^2 \phi_{1\bar{1}}\phi_\tau+8k^3\partial^3 \phi_{1\bar{1}}\\
 \frac{d\phi_{\bar{1}2}}{dt}= & \phi_{\bar{1}2} ( -8\phi_\tau^3+8\phi_e\phi_\tau-12 \phi_\tau\partial \phi_\tau +3k\partial \phi_e -4k^2 \partial^2 \phi_\tau),\\
 & +\partial \phi_{\bar{1}2} (10k\phi_e-24k\phi_\tau^2-16 k^2 \partial \phi_\tau)-12 k^2\partial^2 \phi_{\bar{1}2} \phi_\tau -8k^3\partial^3 \phi_{\bar{1}2},\\
 \frac{d\phi_e}{dt}= & -12k \partial H_1 +k^3 \partial \phi_e +24 k \phi_\tau\phi_{\bar{1}2}  \partial \phi_{1\bar{1}}-24 k \phi_\tau  \phi_{1\bar{1}}\partial\phi_{\bar{1}2}.
\end{aligned}
\end{array} \right.
\end{equation}
\end{ex}
  


\begin{thebibliography}{00}


\bibitem {BDHM} N.J. Burroughs, M.F. De Groot, T.J. Hollowood,
J.L.Miramontes, {\em Generalized Drinfel'd-Sokolov hierarchies, II. The Hamiltonian Structures}. Commun.Math.Phys. 153
(1993) 187-215.

\bibitem{BK} B. Bakalov and V. G. Kac, {\em Field algebras}. IMRN 3 (2003) 123-159.

\bibitem{BDK} A. Barakat, A. De Sole, V. Kac, {\em Poisson vertex algebras in the theory of Hamiltonian equations}. Jpn. J. Math. 4 n. 2 (2009) 141-252.


\bibitem{DHM} M.F. De Groot, T.J. Hollowood, J.L.Miramontes, {\em Generalized
Drinfel'd-Sokolov hierarchies}. Commun.Math.Phys. 145 (1992)  57-84.

\bibitem{DK} A. De Sole, V. G. Kac, {\em Finite vs affine
W-algebras}, Jpn. J. Math. 1 (2006) 137-261.

\bibitem{DKV5}   A. De Sole, V. G. Kac, D. Valeri, {\em Adler-Gelfand-Dickey approach to classical W-algebras within the theory of Poisson vertex algebras,} Int. Math. Res. Not. 21 (2015), 11186-11235.

\bibitem{DKV} A. De Sole, V. G. Kac, D. Valeri, {\em Classical W-algebras and generalized Drinfeld-Sokolov bi-Hamiltonian systems within the theory of Poisson vertex algebras,}  Comm. Math. Phys. 323, no. 2 (2013) 663-711

\bibitem{DKV3} A. De Sole, V. G. Kac, D. Valeri, {\em Classical W-algebras and generalized Drinfeld- Sokolov hierarchies for minimal and short nilpotents,} Comm. Math. Phys. 331 (2014), n. 2, 623-676. Erratum in Commun. Math. Phys. 333 (2015), n. 3, 1617-1619.


\bibitem{DKV2} A. De Sole, V. G. Kac, D. Valeri, {\em Classical W-algebras for $gl_N$ and associated integrable Hamiltonian hierarchies}, arXiv:1509.06878.

\bibitem{DKV4} A. De Sole, V. G. Kac, D. Valeri, {\em Double Poisson vertex algebras and non- commutative Hamiltonian equations}, Adv. Math. 281 (2015), 1025-1099.

\bibitem{DKV1} A. De Sole, V. G. Kac, D. Valeri,  {\em Structure of classical (finite and affine) W- algebras,} J. Eur. Math. Soc. 18 (2016), n. 9, 1873-1908.







\bibitem{DS} V.G. Drinfel'd, V.V. Sokolov, {\em Lie algebras and equations of
Korteseg-de Vries Type}. J.Sov.Math. 30 (1984) 1975-2036


\bibitem{FF} B. Feigin, E. Frenkel, {\em Quantization of the Drinfeld-Sokolov reduction}, Phys. Lett.
B 246  no. 1-2 (1990) 75-81.


\bibitem{GG} W. L. Gan, V. Ginzburg, {\em Quantization of Slodowy
slices}, Int. Math. Res. Not. (2002) 243-255.

\bibitem{H} C. Hoyt, Good gradings of basic Lie superalgebras, Isr. J. Math. (2012) 192-251


\bibitem{IK1} T. Inami, H. Kanno, {\it Generalized N=2 super KdV hierarchies: Lie superalgebraic methods and scalar super Lax formalism. In nite analysis Part A, B (Kyoto, 1991),} 419-447, Adv. Ser. Math. Phys. 16, World Sci. Publ., River Edge, NJ, 1992.

\bibitem{IK} T. Inami, H. Kanno {\em Lie Superalgebraic Approach to Super Toda Lattice
and Generalized Super KdV Equations} Commun. Math. Phys. 136 (1991) 519-542.




\bibitem{K} V. Kac, Vertex algebras for beginners, University Lecture Series, AMS, Vol. 10, 1996 (2nd
Ed., AMS, 1998).

\bibitem{KRW} V. G. Kac, S.-S. Roan, M. Wakimoto, {\em Quantum reduction for affine superalgebras}, 
Comm. Math. Phys.  241 (2003)  307�342

\bibitem{KW} V. G. Kac, M. Wakimoto, {\em Quantum reduction
and representation theory of superconformal algebras}. Adv. Math. 185
(2004) 400-458.


\bibitem{KZ} P. Kulish, A. Zeitlin, {\em Super-KdV equation: classical solutions and quantization}, PAMM � Proc. Appl. Math. Mech. 4 (2004) 576-577







\bibitem{Pol} E. Poletaeva, {\em On finite W-algebras for Lie algebras and superalgebras}, Sao Paulo J.Math. Sci. 7, no.1 (2013) 1-32

\bibitem{PS} E. Poletaeva, V. Serganova, {\em On Kostant's theorem for the Lie superalgebra Q(n)},  Adv. Math. 300 (2016), 320-359.



\bibitem{S} U. R. Suh, Ph.D. Thesis, {\em Structure of classical $\WW$-algebras} (2013).

\bibitem{S2} U. R. Suh, {\em Structure of classical affine and classical affine fractional W-algebras},  J. Math. Phys. 56 (2015) 011706.

\bibitem{S3} U. R. Suh,  J. Math. {\em Structures of classical affine W-algebras associated to Lie superalgebras},
Phys. 57 (2016)  021703.



\bibitem{ZS} Y. Zeng, B. Shu, {\em Finite W-superalgebras for basic classical Lie superalgebras}, J. Algebra 438 (2015), 188-234.

\bibitem{V} D. Valeri, {\em Classical W-algebras within the theory of Poisson vertex algebras, Advances in Lie Superalgebras,} Springer INdAM series, vol. 7 (2013), 203-221.

\bibitem{Zhao} L. Zhao, {\em Finite W-superalgebras for Queer Lie superalgebras}, J. Pure Appl. Algebra 218, no. 7, (2014) 1184-1194











\end{thebibliography}
\end{document}